\documentclass[11pt]{article}

\usepackage{fullpage}
\usepackage{amsmath,amsfonts,amsthm,xspace,hyperref,graphicx}
\usepackage{endnotes}
\usepackage{color}
\usepackage{dsfont}
\usepackage{kpfonts}

\usepackage{mathabx}
\usepackage{thm-restate}
\usepackage{authblk}

\usepackage{amssymb,latexsym}

\usepackage{titlesec}

\newtheorem{theorem}{Theorem}[section]
\newtheorem*{theorem*}{Theorem}
\newtheorem{proposition}[theorem]{Proposition}
\newtheorem{lemma}[theorem]{Lemma}
\newtheorem{claim}[theorem]{Claim}
\newtheorem{fact}[theorem]{Fact}

\newtheorem{definition}[theorem]{Definition}
\newtheorem*{definition*}{Definition}
\theoremstyle{definition}

\newtheorem{remark}[theorem]{Remark}

\newcommand{\beq}{\begin{eqnarray}}
\newcommand{\eeq}{\end{eqnarray}}

\titlespacing*{\paragraph}{0pt}{5pt}{3pt}
\titlespacing*{\section}{0pt}{7pt}{4pt}
\titlespacing*{\subsection}{0pt}{2pt}{2pt}
\titlespacing*{\subsubsection}{0pt}{2pt}{2pt}

\newcommand{\ket}[1]{|#1\rangle}
\newcommand{\bra}[1]{\langle#1|}

\newcommand{\Id}{\ensuremath{\mathop{\rm Id}\nolimits}}

\DeclareMathOperator{\Tr}{\mathrm{Tr}}
\newcommand{\C}{\ensuremath{\mathbb{C}}}
\newcommand{\N}{\ensuremath{\mathbb{N}}}

\newcommand{\R}{\ensuremath{\mathbb{R}}}

\DeclareMathOperator*{\Ex}{ \mathbb E}
\newcommand{\wtilde}{\widetilde}

\newcommand{\Xint}{X_{\textit{int}}}
\newcommand{\Yint}{Y_{\textit{int}}}
\newcommand{\Gint}{G_{\textit{int}} }

\newcommand{\rhohalf}{\rho^{1/2}}

\DeclareMathOperator{\poly}{poly}

\DeclareMathOperator{\valst}{\textsc{Val}^*}
\DeclareMathOperator{\val}{\textsc{Val}}

\newcommand{\eps}{\varepsilon}

\titleformat{\subsection}[block]
{\large \bfseries}
{\thesubsection}{7pt}{}[]

\titleformat{\section}[block]
{\Large \bfseries}
{\thesection}{7pt}{}[]
\begin{document}

\title{Parallel repetition via fortification: analytic view and \\ the quantum case}
\author[1]{Mohammad Bavarian \thanks{bavarian@mit.edu.  Work was supported by  National Science Foundation   grants CCF-0939370 and  CCF-1420956.}}

\author[2]{Thomas Vidick\thanks{vidick@cms.caltech.edu. Work was partially supported by the IQIM, an NSF Physics Frontiers Center (NFS Grant PHY-1125565) with support of the Gordon and Betty Moore Foundation (GBMF-12500028).
}}
\author[1]{Henry Yuen\thanks{hyuen@mit.edu. Work was supported by Simons Foundation grant \#360893, and National Science Foundation Grants 1122374 and 1218547. Work partially conducted while visiting the IQIM at Caltech.}}
\affil[1]{Massachusetts Institute of Technology, Cambridge, MA, USA}
\affil[2]{Department of Computing and Mathematical Sciences, California Institute of Technology, Pasadena, CA, USA }
\date{\today}

\maketitle
\begin{abstract}


In a recent work, Moshkovitz [FOCS '14] presented a transformation on two-player games called ``fortification'', and gave an elementary proof of an (exponential decay) parallel repetition theorem for fortified two-player projection games. In this paper, we give an \emph{analytic reformulation} of Moshkovitz's fortification framework, which was originally cast in combinatorial terms. This reformulation allows us to expand the scope of the fortification method to new settings.

First, we show \emph{any} game (not just projection games) can be fortified, and give a simple proof of parallel repetition for general fortified games. Then, we prove parallel repetition and fortification theorems for games with players sharing quantum entanglement, as well as games with more than two players. This gives a new gap amplification method for general games in the quantum and multiplayer settings, which has recently received much interest.

An important component of our work is a variant of the fortification transformation, called ``ordered fortification", that preserves the entangled value of a game. The original fortification of Moshkovitz does not in general preserve the entangled value of a game, and this was a barrier to extending the fortification framework to the quantum setting.

\end{abstract}

\section{Introduction}

A central concept in theoretical computer science and quantum information is that of a \emph{two-player one-round game}. A two-player game $G$ is specified by question sets $X$ and $Y$, answer sets $A$ and $B$, a distribution $\mu$ over pairs of questions, and a verification predicate $V:A \times  B\times X\times Y\rightarrow \{0,1\}$. The game is played between two cooperating (but non-communicating) players and a referee. The referee samples $(x,y)\in X\times Y$ according to $\mu$ and sends $x$ and $y$ to each player, who provide answers $a\in A$ and $b\in B$ respectively. The players win the game if their answers satisfy the predicate $V(a,b,x,y)$. 

Games arise naturally in settings ranging from hardness of approximation~\cite{hastad2001} and interactive proof systems~\cite{ben1988,fortnow1988} to the study of Bell inequalities and non-locality in quantum physics~\cite{chsh_paper, cleve2004consequences}. Depending on the context it is natural to consider players with access to different resources. In particular, in this work we distinguish between the \textsl{classical value} $\val(G)$ of a game, defined as the maximum winning probability of players allowed to produce their answers using private and  shared randomness, and the \textsl{entangled value} $\val^*(G)$, for which the players may use shared entanglement as well .

An important operation on games and the main focus of this work is that of \emph{repeated tensor product} or \emph{parallel repetition}. This operation takes a game $G$ and a parameter $m$ and outputs a new game $G^{\otimes m}$ in which $m$ independent instances of $G$ are simultaneously played with the two players: the referee samples $m$ independent questions $\{(x_i, y_i)\}_{i=1}^{m}$ from $G$, sends $(x_1, \ldots, x_m)$ to the first and $(y_1, \ldots, y_m)$ to the second player, and checks their corresponding answers $(a_1, \ldots, a_m)$ and $(b_1, \ldots, b_m)$ using the predicate $V=\prod_{i=1}^m V(a_i, b_i, x_i, y_i)$. Parallel repetition is often used in complexity theory in order to perform some form of \emph{amplification}, such as amplifying the completeness-soundness gap of a proof system. A fundamental question that arises in this context is how the value of a repeated game $G^{\otimes m}$ relates to the value of the original game $G$ and the number of repetitions $m$.

The behavior of the game value under parallel repetition can in general be quite subtle~\cite{feige2002error, feige2007understanding,  raz2011counterexample}. In a celebrated paper, Raz showed that if $\val(G) < 1$, then $\val(G^{\otimes m})$ goes to $0$ exponentially fast in $m$~\cite{raz1998parallel}. Even with later simplifications and improvements to the proof (e.g.~\cite{holenstein2007parallel, rao2011parallel, braverman_garg}), Raz's parallel repetition theorem remains a substantial technical result.

Recently, Moshkovitz \cite{mosh2014} introduced a simple yet powerful framework for parallel repetition, called \emph{parallel repetition via fortification}. In this framework, a game $G$ is transformed through an operation called ``fortification'' to a new game $G'$. This new game $G'$ is equivalent to $G$ in that $\val(G) = \val(G')$, but then Moshkovitz shows that behavior of the value of \emph{fortified} games under parallel repetition is much simpler than the general case, and avoids many of the subtleties encountered in the general case. The main benefits of fortified games are two-fold: first, their behavior under parallel repetition is much simpler than the general case, and second, all games can be easily fortified. Thus for nearly all intents and purposes, it suffices to focus on the parallel repetition of fortified games.

Despite its attractive features, the fortification framework~\cite{mosh2014} has some limitations; for instance it is only applicable to  the restricted (though very important) setting of classical two-player projection games. 
In this work, we continue the study of the fortification approach to parallel repetition and try to expand its scope to wider classes of games.

Previously, an attempt to address the limitation of fortification to projection games was made in~\cite[Lemma 1.9]{BSVV}, but this was only partially successful.\footnote{The difficulty there was that fortification increases the alphabet size of the game considerably, which in turn increased the additive error of the parallel repetition to the extent that for general two-prover games the whole approach seemed to entirely break down. We get around this difficulty by showing that for the kinds of games (concatenated games) that arise from the fortification procedure it is possible to establish a parallel repetition theorem where the additive error only depends on the alphabet size of the original, rather than the fortified game. See Subsection \ref{sec:techniques_intro} for more on this.} Here, using a slightly modified form  of the framework (see the discussion of our analytic formulation below), we are able to extend fortification to all \emph{classical games} (i.e.~games that are not projection games, and involve any number of players), as well as the challenging setting of \emph{entangled games}.

More precisely, the following is a summary of our main contributions:

\begin{itemize}
\item 
\textbf{Analytic formulation of fortification.} The framework of parallel repetition by fortification was originally cast in combinatorial terms; Moshkovitz's definition of fortified games, which we describe below in detail, involves a guarantee on the value of every sufficiently large rectangular subgame of a game. In our analytic reformulation, fortified games are defined in terms of \emph{substrategies}, which one can think of as randomized strategies for the game where the probability that the players output an answer may be less than $1$. This definition behaves much more ``smoothly'', allowing us to generalize them to the entangled and multiplayer settings. 

\item \textbf{Fortification of general classical games and  games with more than two players.} Next, we show how to fortify a general $k$-player game $G$, for any $k \geq 2$. We show that for any two fortified general classical games $G'$ and $H'$, $\val(G' \otimes H') \approx \val(G') \cdot \val(H')$. Together this implies new gap amplification results for general (as opposed to projection) two-player and multiplayer classical games. 

\item \textbf{An entangled-value preserving variant of concatenation.} A major obstacle in extending the fortification framework to the quantum setting is that concatenation, the main ingredient of the original fortification results, does not in general preserve the entangled value. That is, if $G'$ is the fortification of $G$, it doesn't generally hold that $\valst(G') = \valst(G)$ (even though $\val(G') = \val(G)$). This is problematic for obtaining gap amplification results: if $\valst(G) = 1$, then $\valst(G^{\otimes n}) = 1$, but $\valst({G'}^{\otimes n})$ could be exponentially small! 
 
To resolve this issue, we augment the ordinary concatenation procedure of  \cite{mosh2014} by giving the players some auxiliary advice input (see Definition \ref{def:G_OF_intro}) which helps in keeping the entangled value unchanged. Using this, we define  a variant of the fortification transformation which we call \emph{ordered fortification}. As desired, in addition to preserving the classical value, this transformation also preserves the entangled value, which is essential for the completeness of our gap amplification result.

\item \textbf{Fortification of games with entangled players.} We show that for a general two-player game $G$, its ordered-fortification $G_{OF}$ is a two-player game such that $\val^*(G_{OF}) = \val^*(G)$, and is also quantumly fortified. We then prove that for any two quantumly fortified games $G'$ and $H'$, $\valst(G' \otimes H') \approx \valst(G') \cdot \valst(H')$. Together this implies a new general gap amplification method for entangled two-player games. This (see Theorem \ref{thm:main_fortification}) is the most technically challenging component of this work.


\end{itemize}

Let us note that our extensions of the fortification approach, as described in the last three items above, are ultimately enabled by the analytic viewpoint described in point 1.  In order to describe the main ideas behind these results, we first briefly recall the combinatorial framework from~\cite{mosh2014}.

\paragraph{The combinatorial framework.} Let $G$  be a two-player game with question sets $X, Y$ and acceptance predicate $V$. For $S\subseteq X$ and $T\subseteq Y$, the \emph{subgame} $G_{S\times T}$ is defined as the game where the referee selects $(x,y)\in X\times Y$ according to $\mu$ conditioned on $x\in S, y\in T$ and checks the players' answers according to the same predicate $V$ (the referee accepts automatically if $\mu(S\times T)=0$). A game $G$ is said to be $(\eps,\delta)$-\emph{combinatorially fortified}\footnote{We shall refer to this notion as  combinatorial fortification to distinguish it from the distinct (though closely related) notion of \emph{analytic fortification} we primarily use throughout the paper; see Definition \ref{def:fortification}.} if  
\begin{equation}\label{eq:old_fortified}
  \val(G_{S\times T})\leq \val(G)+\eps,\qquad \forall \, S\subseteq X,  \, T\subseteq Y, \mbox{ s.t. } \mu(S \times T) \geq \delta.
\end{equation}
The main insight underlying \cite{mosh2014} is that  games satisfying (\ref{eq:old_fortified}) also satisfy a strong form of parallel repetition (up to some number of rounds depending on $\eps$, $\delta$, and the alphabet size of $G$).  This motivates the following approach to parallel repetition: Given a game $G$, Moshkovitz transforms the game $G\rightarrow G'$ such that $\val(G')\approx\val(G)$ and $G'$ is $(\eps, \delta)$-combinatorially fortified for an appropriate choice of $(\eps, \delta)$. Since fortified games satisfy a strong form of parallel repetition, one expects 
\begin{equation}
	\val( {G'}^{\otimes m})\approx \val(G')^m\approx \val(G)^m.
\end{equation}
Indeed, by appropriately choosing the parameters $(\eps, \delta)$,~ \cite{mosh2014} can show that the full procedure 
\begin{equation}\label{eq:fort_diagram}
G\longrightarrow G' \longrightarrow {G'}^{\otimes m}	
\end{equation}
amounts to a size-efficient method of \emph{gap amplification}. That is, we have
\begin{equation}\label{eq:gap_amplification}
\begin{array}{ccc}
\val(G)\geq c \quad \Rightarrow \quad &\val({G'}^{\otimes m})\gtrsim c^m \\
	\val(G)\leq s \quad \Rightarrow \quad &\val({G'}^{\otimes m}) \lesssim s^m
\end{array},
\end{equation}
where we refer to the first condition as completeness and the second as soundness.  The gap amplification procedure  of Moshkovitz $G\rightarrow G'\rightarrow {G'}^{\otimes m}$ from \eqref{eq:fort_diagram} has three components: (i) a preprocessing step (biregularization), (ii) fortification, (iii) parallel repetition for fortified games. 

The goal of the preprocessing step --  the simplest step of the three --  is to make the game \emph{biregular} (a game $G$ is called biregular if the marginals of questions on both Alice and Bob sides are uniform), since it is typically easier to analyze the fortification procedure for such games.  The second step is \emph{fortification}, which is the main technical ingredient of the whole approach. It is achieved by ``concatenating'' the game (see Section~\ref{sec:techniques_intro} below) with appropriate bipartite pseudorandom graphs.  The third step $G'\rightarrow {G'}^{\otimes m}$ is the parallel repetition of fortified games, which as observed by \cite{mosh2014} is considerably simpler to analyze than the general (non-fortified) games.


\subsection{Results and techniques}\label{sec:techniques_intro}

The main result of our work is the extension of the fortification framework to general classical games (with any number of players) and two-player entangled games. On the way to these results we prove new results on all three components of the fortification framework: (i) biregularization, (ii) fortification, and (iii) parallel repetition. In this subsection, we discuss some of these results in detail. 


\paragraph{Parallel repetition.}
A main contribution of \cite{mosh2014} was the realization that the two-player projection fortified games satisfy a strong form of parallel repetition, up to an additive error. This additive error depended on the parameters of fortification as well as the alphabet size of the fortified game. In this work, we prove an improved parallel repetition theorem (Theorem \ref{thm:PR2_multiround}) which has the same dependance in the parameters of fortification, but instead of the alphabet size of the resulting fortified game, it only depends on the alphabet size of the original game (which  has exponentially smaller alphabet size). This new parallel repetition theorem is crucial for extending the fortification framework to the setting of general (as opposed to projection) two-player games. 

Let us remark that the reason why alphabet blow-up of fortification does not cause an issue for projection games is because for projection games it suffices to only fortify one side of the game (by working with so-called ``square projection" version of the game). As a result there is no alphabet blow-up for the ``unfortified'' side, which allows the arguments of \cite{mosh2014,BSVV} to go through. This one-sided fortification does not work for general games, which is why we need Theorem \ref{thm:PR2_multiround}.

\paragraph{Fortification.} 
We start with a definition. Let  $G = (X\times Y,A\times B,\mu,V)$ be a game, and $M$ and $P$  two bipartite graphs over vertex sets $(X',X)$ and $(Y',Y)$ respectively. For each $x\in X$ or $x'\in X'$ let $N(x) \subseteq X'$ and $N(x')\subseteq X$ denote the set of neighbors of $x$ and $x'$, respectively (similarly for any $y, y'$).
\begin{definition}[Concatenated game~\cite{mosh2014}]\label{def:concat_intro}
In the concatenated game $G'=(M\circ G\circ P)$, the referee selects questions $(x,y)$ according to $\mu$, and independently selects a random neighbor $x'$ for $x$ using $M$, and $y'$ for $y$ using $P$. The players receive questions $x'$ and $y'$ and respond with assignments $a':N(x')\rightarrow A$ and $b':N(y')\rightarrow B$ respectively. The players win if $V(a'(x), b'(y),x,y)=1$. 
\end{definition}
Our first two main results show how, both in the classical and quantum settings, any game can be fortified by concatenating it with bipartite graphs $M$ and $P$ with sufficiently good spectral expansion.\footnote{Even though explaining these results in full requires the definition of analytically fortified games, which we introduce only later in Section~\ref{sec:anal_fort}, the analogy with the notion of combinatorially fortified games from \eqref{eq:old_fortified} should still be sufficient to understand the basic ideas.} (See Section~\ref{sec:expanders} for the definition of spectral expanders, and Section~\ref{sec:anal_fort} for the notion of weak fortification.)

\begin{theorem}[Main-classical]\label{thm:classical_fortification}
Let $G$ be a biregular game and $M$ and $P$ two bipartite $\lambda$-spectral expanders. If $\lambda\leq \frac{\eps}{2}\sqrt{\frac{\delta}{2}}$, then  the concatenated game $G'=(M\circ G\circ P)$ is $(\eps,\delta)$-weakly fortified against classical substrategies.
\end{theorem}

\begin{theorem}[Main-quantum]\label{thm:main_fortification}
Let $G$ be a biregular game and $M$ and $P$ two bipartite $\lambda$-spectral expanders. If $\lambda\leq \frac{\eps^2\delta}{56}$, then  the concatenated game $G'=(M\circ G\circ P)$ is $(\eps,\delta)$-weakly fortified against entangled strategies. 
\end{theorem}

We stress that  both in the quantum and classical settings the procedure used to fortify a game is precisely the same, i.e.~concatenation with spectral expanders, and the only difference is in the resulting parameters. Despite the similarities, the proof of Theorem \ref{thm:main_fortification} is significantly more involved, requiring several new ideas and substantial matrix analytic arguements.

Next we discuss a distinctively quantum phenomenon which makes the construction of a full quantum gap amplification theorem -- quantum analogue of \eqref{eq:gap_amplification} -- considerably more difficult. As it turns out, even though Theorem \ref{thm:main_fortification} is sufficient to prove the soundness case of the gap amplification theorem,  the concatenation procedure used in the process can undermine the completeness condition (i.e.~$\valst({G'}^{\otimes m}) \gtrsim \valst(G)^m$ in general fails to hold).

The issue is as follows: let $G$ be a game and $G'=(M\circ G\circ P)$ be a concatenated version of $G$. Classically we have $\val(G')=\val(G)$. Quantumly, even though we still have $\valst(G')\leq \valst(G)$ the other direction in general fails: we would have liked to argue that the players in $G'$ are able to utilize the strategy in $G$ to achieve the same success probability in the concatenated game, but this seems impossible: having received $x'\in X'$ and $y'\in Y'$, the players have access to lists $N(x')\subseteq X$ and $N(y')\subseteq Y$ that they know contain the true questions of the referee, i.e.~$x^*\in N(x')$, $y^*\in N(y')$. The players would like to apply their optimal strategy in $G$ to each and every $(x,y)\in N(x')\times N(y')$ simultaneously, but this is in general impossible in the quantum setting.\footnote{This is because the measurement operators of different questions do not in general commute which prevents Alice (say) to obtain simultaneous answers for all questions in $N(x')$. As a further illustration of this issue, see Section \ref{sec:concat_prelim} for an example of a game where $\valst(G')< \valst(G)$. }

Note that the same issue does not arises classically because the optimal strategy in $G$ can be taken to be a deterministic one, and the players in $G'$ can use the same labeling suggested by the optimal strategy in $G$ to give labels to all of $N(x')$ and $N(y')$ simultaneously. This strategy however relies on the fact that classically different questions have a  simultaneous labeling, a fact which certainly has no quantum analogue.

 We resolve the above issue using a novel entangled value-preserving variant of fortification which we call \emph{ordered fortification}. The basic idea for ordered fortification is to give the players some extra advice information which helps in preserving the entangled value. 
 
  Let $G$  be a game and $G'=(M\circ G\circ P)$ be a concatenated version of $G$. There is an extra parameter $l$ in the construction defined as $l=\max\left\{ \max_{x'\in X'} |N(x')|, \max_{y'\in Y'} |N(y')| \right\}$.
\begin{definition}[Ordered concatenation]\label{def:G_OF_intro}
Let $G$ and  $G'$ be as above. In $G'_{OF}$, the referee samples $(x,y)$ according to $G$ and picks random neighbors $x'\sim N(x)$ and $y'\sim N(y)$ independently. She then also picks two random injective maps $r_{x'}:N(x')\rightarrow [l]$ and $s_{y'}:N(y')\rightarrow [l]$ conditioned on $s_{x'}(x)= r_{y'}(y)$. The referee sends $x'$ and $r_{x'}$ to the first player, and $y'$ and $s_{y'}$ to the second and accepts if the players' answers $a':N(x')\rightarrow A$ and $b':N(y')\rightarrow B$ satisfy  $V(a'(x), b'(y), x,y)=1$.
\end{definition}
Here the crucial point is that  $r_{x'}$ and $s_{y'}$ are correlated. They give  matching labels to true questions $x$ and $y$. To achieve the same winning probability as in $G$, the players in $G'_{OF}$ will share $l$ copies of the state $|\psi\rangle$ from the optimal strategy in $G$. For each $x^*\in N(x')$ with label $i=r_{x'}(x^*)$, the first player will apply the optimal $G$-strategy for $x$ to the $i^{\textit{th}}$ copy of $|\psi\rangle$ (similarly for the second player). The fact that $r_{x'}(x)= s_{y'}(y)$ ensures that for the true questions $x$ and $y$ the players apply the optimal $G$ strategies to the same copy of $|\psi\rangle$, and hence are able to win with exactly the same winning probability as in $G$. 

Of course, the crucial part here is that even though the auxiliary information in $r_{x'}$ and $s_{y'}$ is helpful to the players for replicating the winning probability of  $G$, it should not be ``too helpful". In particular, we need to still be able to prove that $G'_{OF}$ is fortified with appropriate parameters. This point is established by the following theorem.

\begin{theorem}[Main-ordered fortification]\label{thm:fort_OF_main}
Let $G$ be a game and $M$ and $P$ be two bipartite graphs as above. Let $G'_{OF}$ be constructed from $G$ and $G'=(M\circ G\circ P)$ as in Definition \ref{def:G_OF_intro}. Then, we have 
\[ \valst(G'_{OF})= \valst(G).\]
Furthermore if $M$ and $P$ are $\lambda$-spectral expanders and $\lambda\leq \frac{\eps^2\delta}{56}$, then $G'_{OF}$ is also $(\eps, \delta)$ weakly fortified.
\end{theorem}
We prove Theorem \ref{thm:fort_OF_main} in Section \ref{sec:modified_fortification} using a spectral argument that reduces it to Theorem \ref{thm:main_fortification}. Beside the above, we also prove a simple multiplayer fortification in Section \ref{sec:classical} for classical games. It may be possible to adapt the proofs of Theorem~\ref{thm:main_fortification} and Theorem \ref{thm:multiplyaer_classical} to obtain a \emph{multiplayer fortification theorem} for entangled games. Although plausible, some further technical issues arise in this case which we do not pursue here. 

\paragraph{Biregularization.} As already mentioned, biregularization is a minor (but necessary) step in the fortification framework. Our biregularization lemmas are presented in Subsection \ref{sec:bireg_intro} and are proved in Appendix \ref{app:bireg}. In terms of final statement, our biregularization lemmas are incomparable with those of \cite{mosh2014,BSVV}. For example, in the case of graphical games, we prove a biregularization lemma which preserves the value exactly but has a cubic blow-up in the number of questions, whereas the biregularization lemmas from~\cite{BSVV, mosh2014} had a nearly linear blow-up but only preserved value up to an additive error. Moreover, in this work we prove biregularization for all games whereas \cite{BSVV, mosh2014} only considered graphical games. (See Subsection \ref{sec:bireg_intro} for definitions.)


\subsection{Related work}\label{sec:related_work}
The main result underlying the present work is Moshkovitz~\cite{mosh2014}, where the framework of parallel repetition via fortification was first introduced. Some simplifications and corrections to the work of Moshkovitz appeared in Bhangale et al.~\cite{BSVV}. In particular, an important contribution of \cite{BSVV} was the clarification of the best bounds possible in classical fortification theorems~\cite[Appendix C]{BSVV}. Going back, the general idea of modifying the game in order to facilitate  its analysis under parallel repetition originates from the work of Feige and Kilian~\cite{feige2000two} who introduced the confuse/miss-match style repetition of games. The Feige-Kilian type parallel repetition was later extended by Kempe and Vidick~\cite{kempe2011parallel} to the quantum setting allowing them to obtain the first general parallel repetition theorem for quantum games.\footnote{Their transformation did not preserve the entangled value, and hence did not lead to a fully general hardness amplification method for entangled games (being restricted to the case where the completeness holds with classical strategies). Our Theorem~\ref{thm:main_fortification}, without the improvement of Theorem~\ref{thm:fort_OF_main}, is applicable to a similar setting.} 

Another important set of ideas underlying our work is related to the analytic approach to parallel repetition pioneered by Dinur and Steurer \cite{DS14}, further extended by Dinur et al.~\cite{DinurSV14}. Our analytic reformulation of fortification framework is very much inspired by the ideas in these works.

Yet another different stream of work (more distantly related to the present work) follows the original ideas of Raz and Holenstein \cite{raz1998parallel, holenstein2007parallel} by taking a more information theoretic approach to quantum parallel repetition. The first results in this direction were obtained by Chailloux and Scarpa~\cite{ChaillouxS14} and  Jain et al.~\cite{JainPY14} who prove exponential-decay parallel repetition results for free two-player games. Their analysis, as well as the follow-up work of Chung et al.~\cite{chung2015parallel}, provided the basis for the recent work of the authors~\cite{BVY15} who obtained hardness amplification method in  great generality by introducing and analyzing the parallel repetition of a class of games called \emph{anchored games}. The final hardness amplification results obtained here through fortification are incomparable to that of \cite{BVY15}: fortification allows for a somewhat faster rate of decay in some regimes,\footnote{For example, for general two-player games the fortification approach gives a nearly perfect decay in terms of number of repetitions and question size blow-up, whereas~\cite{BVY15} and other information theoretic results have a $(1-\eps^2)^{\Omega(m)}$ type behavior which is weaker than fortification when $m\ll \eps^{-1}$.} yet it suffers from a much larger blow-up in terms of alphabet size.


Turning to the multiplayer setting, very little was known prior to the present work and \cite{BVY15}. It is folklore that free games with any number of players satisfy a parallel repetition theorem, and this was explicitly proved in both classical and quantum settings in~\cite{chung2015parallel}. Multiplayer parallel repetition has been studied in the setting of \emph{non-signaling strategies}, a superset of entangled strategies which allows the players to generate any correlations that do not imply communication. Buhrman et al.~\cite{buhrman2013parallel} show that the non-signaling value of a game $G$ with any number of players decays exponentially under parallel repetition, with a rate of decay that depends on the entire description of the game $G$. Arnon-Friedman et al.~\cite{arnon2014non} and Lancien and Winter~\cite{lancien2015parallel} achieve similar results using a different technique based on ``de Finetti reductions''. An interesting fact about the latter work~\cite{lancien2015parallel} is the use of a notion of sub-no-signalling strategies which seems related to our notion of quantum/classical substrategies.
\subsection{Organization}
In Section \ref{sec:prelim}, we introduce some basic definitions and notation including the notion of substrategies, induced strategies, and  some other basic results and definitions that are used throughout the paper. In Section \ref{sec:fort_frame}, we  complete the presentation of our main results (which was started in the introduction), discuss the parameters of the final gap amplification results, and present the formal definition of analytically fortified games. 

The remaining sections contain proof of the main theorems. Our parallel repetition theorem is proved in Section \ref{sec:concat}. Theorems \ref{thm:classical_fortification} and \ref{thm:main_fortification} are proved in Sections \ref{sec:classical} and  \ref{sec:quantum_generalization}, respectively. The reduction from Theorem \ref{thm:fort_OF_main} to Theorem \ref{thm:main_fortification} is given in Section \ref{sec:modified_fortification}.  The biregularization lemmas are proved in Appendix \ref{app:bireg}. We conclude by some open problems in Section \ref{sec:open_problems}.


\section{Preliminaries}\label{sec:prelim}
Given a distribution $\mu$,  by $z\sim \mu$ we mean that the random variable $z$ is distributed according to $\mu$. For a set $S$, by $z\sim S$ we mean $z\sim U_S$ where $U_S$ is the uniform distribution over $S$. For Hermitian matrices $A,B$ we write $A \geq B$ if and only if $A-B$ is positive semidefinite. 

Given two games  $G_1$ and $G_2$, we define the tensor product game $G_1\otimes G_2$  as the game where the referee selects two pairs of questions  $(x_1,y_1)$ and $(x_2, y_2)$ independently according to $G_1$ and $G_2$, sends $(x_1, x_2)$ and $(y_1,y_2)$ to Alice and Bob respectively, and checks their answers $(a_1, a'_2)$ and $(b_1, b_2)$ according to the product predicate $\prod_{i=1}^{2} V(a_i,b_i, x_i, y_i)$. 

\subsection{Game value and strategies}
The main goal of this section is to introduce the notion of classical and quantum \emph{substrategies} which replace the notion of \emph{subgames} from \cite{mosh2014, BSVV}. As subgames were central in the \emph{combinatorial} framework of~\cite{mosh2014}, substrategies are similarly central to our analytic framework.

Let $G$ be a game with question sets $X,Y$, answer sets $A,B$, predicate $V$, and question distribution $\mu$ on $X\times Y$. 
\begin{definition}[Classical substrategies]\label{def:classical_sub}
Let $G = (X\times Y,A \times B, \mu,V)$ be a two-player game. A \emph{classical substrategy} is given by $(f,g)$ where $f:X\times A\rightarrow [0,1], g:Y\times B\rightarrow [0,1]$ satisfy
\begin{equation*} 
\forall x\in X, \, f(x):= \sum_{a} f(x,a)\leq 1, \qquad   \forall y\in Y,\, g(y):=  \sum_{b} g(y,b)\leq 1.
\end{equation*}
We call $(f,g)$ a  \emph{``complete strategy"} (sometimes simply \emph{strategy}) if equality holds in all above inequalities, i.e. $f(x)=g(y)=1$  for all $x,y$.
\end{definition}
\begin{definition} Given a substrategy $(f,g)$,  the value of $G$ with respect to $(f,g)$ is given by 
\begin{equation} \label{eq:val_general_fcn} \val(G,f,g):=\Ex_{(x,y)\sim \mu} \sum_{a\in A, b\in B} V(a,b,x,y) \, f(x,a)\cdot g(y,b).	
\end{equation}
The \emph{classical value} of $G$ is 
 \begin{equation}\label{eq:val_anal}
	 \val(G):= \sup_{f,g}\, \val(G,f,g),
\end{equation}
where the supremum is taken over all complete strategies $f,g$. 
\end{definition} 
We note that the definition given by \eqref{eq:val_anal} can be easily seen to be equivalent to the more traditional definition of the classical value, i.e.~
\begin{equation}\label{eq:val_comb}
	 \val(G):= \max_{\substack{p:X\rightarrow A \\ q:Y\rightarrow B}} \,\Ex_{(x,y)\sim \mu} V(p(x), q(y), x,y),
\end{equation}
because any strategy $f:X\times A\rightarrow [0,1], g:Y\times B\rightarrow [0,1]$ can be written as convex combination of a collection of strategies of $\{0,1\}$ valued strategies; on the other hand, taking supremum over $f,g$ which are $\{0,1\}$ valued is precisely equivalent to \eqref{eq:val_comb}.

Next, we extend the above notions to the quantum setting. 

\begin{definition}[Quantum substrategies]
	Let $G = (X \times Y, A \times B,\mu,V)$ be a two-player game. A \emph{quantum} (or \emph{entangled}) \emph{substrategy} for $G$ is a tuple $(\ket{\psi},\{A_x^a\},\{B_y^b\})$ defined by an integer $d\in \N$, a unit vector $|\psi\rangle \in \C^{d\times d}$ and sets of positive semi-definite matrices $\{A_{x}^{a}\}_{x\in X, a\in A}, \{B_{y}^b\}_{y\in Y, b\in B}$ over $\C^d$ satisfying
\begin{equation}
	 \forall x\in X, \,  A_x:=\sum_{a} A_{x}^a \leq \Id, \qquad   \forall y\in Y,\, B_y:= \sum_{b} B_y^b\leq \Id.
\end{equation}
 If $A_x = B_y = \Id$ for every $x,y$ the quantum substrategy is called a \emph{``complete strategy"} (sometimes simply \emph{strategy}). 
\end{definition}
\begin{definition}
Given  a quantum substrategy $(\ket{\psi},\{A_x^a\},\{B_y^b\})$,  the  value of $G$ with respect to $(\ket{\psi},\{A_x^a\},\{B_y^b\})$ is given by
\[
	\valst(G,\ket{\psi},\{A_{x}^{a}\},\{B_{y}^b\}) = \Ex_{(x,y) \sim \mu} \sum_{a,b} V(a,b,x,y) \bra{\psi} A_x^a \otimes B_y^b \ket{\psi}.
\]
The \emph{entangled value} of $G$ is defined as
\begin{equation}
\valst(G)= \sup_{|\psi\rangle, \{A_x^{a}\}, \{B_y^b\}} \valst(G, |\psi\rangle, \{A_x^a\}, \{B_y^b\}),	
\end{equation}
where the supremum is taken over all complete strategies $(\ket{\psi},\{A_x^a\},\{B_y^b\})$.
\end{definition}

\subsection{Concatenated games}\label{sec:concat_prelim}

Let  $G = (X\times Y,A\times B,\mu,V)$ be a game, and $M$ and $P$  two bipartite graphs over vertex sets $(X',X)$ and $(Y',Y)$ respectively. For each $x\in X$ or $x'\in X'$ let $N(x) \subseteq X'$ and $N(x')\subseteq X$ denote the set of neighbors of $x$ and $x'$, respectively (similarly for any $y, y'$). Recall the definition of Concatenated Games from the introduction.
\begin{definition*}[Definition \ref{def:concat_intro} restated]
In the concatenated game $G'=(M\circ G\circ P)$, the referee selects questions $(x,y)$ according to $\mu$, and independently selects a random neighbor $x'$ for $x$ using $M$, and $y'$ for $y$ using $P$. The players receive questions $x'$ and $y'$ and respond by assignments $a':N(x')\rightarrow A$ and $b':N(y')\rightarrow B$ respectively. The players win if $V(a'(x), b'(y),x,y)=1$. 

For a concatenated game $G'=(M\circ G\circ P)$, we refer to $G'$ as the \emph{outer game} and to $G$ as the \emph{inner game}.
\end{definition*}

Let $G'=(M\circ G\circ P)$ be a concatenated game. Let $d_{X'}=\max_{x'\in X'} |N(x')|$, $d_{Y'}=\max_{y'\in B'}|N(y')|$. Then,  the alphabet  of the concatenated game is given by  $A'= A^{d_{X'}}$, $B'= B^{d_{Y'}}$.  Similarly, it is easy to see that the distribution  $\mu'$ of questions in $G'$ is given by
$\mu'(x',y')= \Ex_{(x,y)\sim \mu} \frac{1_{x'\in N(x)} }{|N(x)|}\cdot \frac{1_{y'\in N(y)}}{ |N(y)|}.$

\begin{definition}
Let $G'=(M\circ G\circ P)$ be a concatenated game. To any pair of substrategies $(f,g)$ for $G'$ we associate the \emph{induced substrategy}\footnote{Note the slight (but convenient) abuse of notation due to the use of the same letter to represent a substrategy and the corresponding induced substrategy. The more accurate but  more cumbersome way of denoting the induced strategies in in~\cite{DS14}'s language would have been $Mf$ and $Pg$.}
\begin{equation}
f(x,a):= \Ex_{x'\sim N(x)} \sum_{a': a'(x)= a} f(x',a'), \qquad 	g(y,b):= \Ex_{y'\sim N(y)} \sum_{b': b'(y)= b} f(y',b').
\end{equation}
Similarly, given an entangled substrategy $(\ket{\psi},\{A_{x'}^{a'}\},\{B_{y'}^{b'}\})$ for  $G'$, we define the \emph{induced substrategy} as 
\begin{equation} A_{x}^a := \Ex_{x'\sim N(x)} \sum_{a'(x)= a} A_{x'}^{a'}, \qquad 	B_y^b:= \Ex_{y'\sim N(y)} \sum_{b'(y)= b} B_{y'}^{b'}.
\label{eq:projected_operators}
\end{equation}
\end{definition}

Intuitively, an induced strategy is a strategy for the inner game in which the players proceed as follows:  given question $x\in X$, $y\in Y$ and a strategy $(f,g)$ for the outer game,  the players select two random neighbors of their questions $x'\in N(x), y'\in N(y)$ independently, and play according to the labeling of $x,y$ suggested by  $(f,g)$ at $x'$ and $y'$.  

The following simple proposition will play an important role throughout the paper.
\begin{proposition}\label{prop:val_inequality_classical_quantum}
	Let $G'=(M\circ G\circ P)$ be a concatenated game.	The value of any classical strategy $(f,g)$ (resp. quantum strategy $(\ket{\psi},\{A_{x'}^{a'}\},\{B_{y'}^{b'}\})$) in the outer game $G'$ is equal to the value of the induced strategy  in the inner game $G$:
	\begin{equation}\label{eq:induced-val}
		\val(G',f,g) = 	\val(G,f,g)\qquad\text{and}\qquad	\valst(G',\ket{\psi},\{A_{x'}^{a'}\},\{B_{y'}^{b'}\}) = \valst(G,\ket{\psi},\{A_x^a\},\{B_y^b\}).
\end{equation}
As a consequence, 
	\begin{equation}\label{eq:induced-val-lower}
	\val(G')\leq \val(G), \qquad \text{and}\qquad \valst(G')\leq \valst(G).
	\end{equation}
	Furthermore, 
\begin{equation}\label{eq:induced-val-class} \val(G')=\val(G).
\end{equation}
\end{proposition}

\begin{proof}
The first equality in~\eqref{eq:induced-val} follows from linearity of expectation and the definition of induced strategies as
\begin{align*}
	\val(G',f,g) &=\Ex_{(x,y)\sim \mu}\Ex_{x'\sim N(x)} \Ex_{y'\sim N(y)} \sum_{a'\in A', b'\in B'} V(a'(x), b'(y), x,y) \, f(x',a')\cdot g(y',b')\\
&=\Ex_{(x,y)\sim \mu}  \sum_{a\in A, b\in B} V(a, b, x,y) \, f(x,a)\cdot g(y,b)\\
&=	\val(G,f,g).
\end{align*}
The second equality is proved similarly. 
The two inequalities~\eqref{eq:induced-val-lower} follow directly from~\eqref{eq:induced-val}. To show~\eqref{eq:induced-val-class} it remains to show that $\val(G')\geq \val(G)$. Consider an optimal deterministic strategy for $G$ 	given by $p:X\rightarrow A$ and $q:Y\rightarrow B$. For any $x'\in X'$, $y'\in Y'$  define $a':N(x')\rightarrow A$ according to $p$ and $b':N(y')\rightarrow B$ according to $q$. It is easy to see that this achieves the same value in $G'$ as $(p,q)$ did in $G$. 
\end{proof}
As mentioned in the introduction, the quantum analogue of~\eqref{eq:induced-val-class} does not hold in general. For example, consider the case that  $M$ and $P$ are complete bipartite graphs. In this case, the players playing $G'=(M\circ G \circ P)$ need to provide a labeling to all vertices in $X$ and $Y$ simultaneously. But this is essentially just a classical strategy as the labelings for $X,Y$ are now fixed. Hence, $\valst(G')=\val(G)$, the classical value, which in many cases could be much smaller than $\valst(G)$.

 \subsection{Biregularization}\label{sec:bireg_intro}
As in \cite{mosh2014, BSVV} we prove our fortification theorems for the special class of \emph{biregular games}.

 \begin{definition}
A two-prover game $G=(X\times Y,A\times B,\mu,V)$ is called biregular if the marginals of $\mu$ on $X$ and $Y$ are both uniform. 
\end{definition}

The following lemma justifies that for our purposes we may always assume a game is biregular.

\begin{lemma}[Biregularization lemma]\label{lem:biregular}
Let $G=(X\times Y,A\times B,\mu,V)$ be a two-prover game and $\tau\in (0,1)$ a fixed constant. There exists an efficient algorithm that given $G$ produces a biregular game  $\Gint$ with question sets $\Xint$ and $\Yint$ of cardinality at most 
\begin{equation}
|\Xint |\leq \frac{ 8  |X|^2   |Y|}{\tau}, \qquad |\Yint |\leq \frac{8 |X||Y|^2}{\tau},	
\end{equation}
the same answer alphabet size as $G$, and value satisfying 
\begin{equation}\label{eq:value_bireg}
	\val(G)\leq \val(G_{int})\leq \val(G)+\tau, \qquad \valst(G)\leq \valst(G_{int}) \leq \valst(G)+\tau.
\end{equation}
\end{lemma}

Note that \eqref{eq:value_bireg} implies that applying the Biregularization Lemma to a game never decreases  its value, and hence the procedure is completeness preserving. 

A widely used class of games in applications are so-called \emph{graphical games}, for which we can get an improved biregularization result that does not require any approximation factor $\tau$.  

\begin{definition}\label{def:graphical_games}
A graphical game $G$ is a game where the questions are given by choosing an edge of a bipartite graph uniformly at random (i.e.~$E\subseteq X\times Y$ and $\mu(x,y)=\frac{1}{|E|}$ if $(x,y)\in E$ and $\mu(x,y)=0$ otherwise). The predicate and the answers do not have any restrictions.
\end{definition}
\begin{lemma}[Biregularization lemma, graphical case]\label{lem:biregular_graphical}
Suppose $G$ is two-prover graphical game with $E$ edges between $(X,Y)$. There exists an efficient algorithm that given $G$ produces a biregular game  $\Gint$ with question sets $\Xint$ and $\Yint$ bounded by 
\begin{equation}\label{eq:prelim_1}
	 |\Xint |\leq |E| \cdot |X|\leq |X|^2 |Y|, \qquad |\Yint |\leq |E|\cdot |Y|\leq |X| |Y|^2,
\end{equation}
the same answer alphabet size as $G$, and the value satisfying 
\begin{equation}
	\val(G)= \val(G_{int}) \qquad \valst(G)= \valst(G_{int}).
	\end{equation}
\end{lemma}
\begin{remark}In the above, we can allow for multiple edges across vertices of $G$. In this case $E$ must be taken as a multi-set and the bound $ |E| \leq |X||Y|$ used in~\eqref{eq:prelim_1} must be suitably modified.
\end{remark}

Interestingly, our technique for proving the biregularization lemmas is concatenation itself! This is done in Appendix \ref{app:bireg}.

\subsection{Expanders}\label{sec:expanders}

The method used in~\cite{mosh2014, BSVV}  for fortifying a game is concatenation with sufficient pseudorandom bipartite graphs. This is done using  extractors in \cite{mosh2014} whereas expanders are employed in \cite{BSVV}.\footnote{The two approaches however lead to essentially to similar parameters (e.g.~$\lambda=O(\eps\sqrt{\delta})$ to get $(\eps,\delta)$-fortified graph where $\lambda$ is the second largest singular value of normalized adjacency matrix of the concatenating graph.); moreover, in the classical setting the approaches are in fact are more or less equivalent. See \cite{BSVV} for more.} Here we  follow the latter approach and use expanders. 

Let $M = (X' \times X, E)$ be a bipartite graph. For $x\in X$ let $N(x)\subseteq X'$ denote the set of neighbors of $x$ and similarly for $x'\in X'$. We shall work with graphs that are $X$-regular, i.e.~$d=|N(x)|$ for all $x\in X$. Define distributions $\mu$ and $\mu'$ on $X$ and $X'$ via
\[ \mu(x)=\frac{1}{|X|}, \qquad \mu'(x')= \frac{|N(x')|}{d}.\]
for all $x\in X$ and $x'\in X'$. Note that $\mu'(x')$ is the probability of obtaining $x'$ by sampling $x\sim \mu$ and taking a random neighbor of (according $M$) $x$ . Let $\mathcal M$ be the following normalized adjacency matrix of $M$
$$
	\mathcal M(x,x') = 
\left\{
	\begin{array}{ll}
		\frac{1}{d} \cdot \sqrt{\frac{\mu(x)}{\mu'(x')}}  & \mbox{if } x' \in N(x) \\
			0		& \mbox{otherwise}
	\end{array}
\right.
$$
We usually view $\mathcal M$ as an operator from $\ell_2(X')$ to $\ell_2(X)$. Note that when $M$ is a biregular expander we get the simpler definition $\mathcal M(x,x') = \frac{1}{d} \sqrt{\frac{|X'|}{|X|}}$ for $x'\in N(x)$, and $0$ otherwise.
\begin{definition}
A bipartite graph $M$ is called a \emph{$\lambda$-spectral expander} if the second-largest singular value of $\mathcal M$ is at most $\lambda$. 
\end{definition}
A simple useful proposition for us is the following:
\begin{proposition}
\label{prop:expander}
	Let  $M = (X' \times X, E)$ be a bipartite $\lambda$-spectral expander. For $f: X' \to \R$ and $x\in X$ let $f(x) = \Ex_{x' \sim N(x)} f(x')$, and $\bar{f} =\Ex_{x'\sim \mu'} f(x')= \Ex_{x \sim \mu} f(x)$. Then
\begin{equation}\label{eq:expander_prop}
		\Ex_{x \sim \mu } (f(x) - \bar{f} )^2 \leq \lambda^2 \Ex_{x'\sim \mu'} (f(x') - \bar{f})^2.
\end{equation}

\end{proposition}
\begin{proof}
Let $p_\mu \in \R^{X}, p_{\mu'}\in \R^{X'}$  to unit vectors defined as $p_{\mu}(x):=\sqrt{\mu(x)}$ and $p_{\mu'}(x'):=\sqrt{\mu(x')}$.  Let$q_X(x):= \sqrt{\mu(x)} \, f(x)$, $q_{X'}(x'):= \sqrt{\mu(x')}\, f(x')$.

First, observe that $\mathcal M \, p_{\mu'} = p_{\mu}$ and $\mathcal M^{t} \, p_{\mu'}= p_{\mu}$. It follows that $(p_{\mu}, p_{\mu'})$ form a pair of singular vectors of $\mathcal M$. Moreover, it is easy to see\footnote{e.g. by appealing to the Perron-Frobenius theorem.} that these are  top singular vectors which shows that $\| \mathcal M\|_{op}=1$. Now notice that 
\begin{equation}
\Ex_{x'\leftarrow \mu'} (f(x')-\bar{f})^2= \sum_{x'} (\sqrt{\mu(x')} f(x')- \bar{f} \sqrt{\mu(x')})^2= \|q_{X'}- \bar{f} p_{\mu'}\|_{2}^2,
\end{equation}
Second, observe
\begin{equation}
	\mathcal M \, q_{X'}= q_{X}.
\end{equation}
As such, \eqref{eq:expander_prop} precisely corresponds to 
\begin{equation}
\|q_X- \bar{f} \, p_{\mu} \|_2^2= \| \mathcal M\, (q_{X'}-\bar{f} \, p_{\mu'})	\|_{2}^2 \leq \lambda^2\cdot \| q_{X'}-\bar{f} \, p_{\mu'}\|_2^2. 
\end{equation}
The claim follows by noting the orthogonality property
\begin{equation}
\langle q_{X'}- \bar{f} \, p_{\mu'}, p_{\mu'}\rangle = \sum_{x'} \mu(x') f(x')- \bar{f}=0.
\end{equation}
\end{proof}


\section{Fortification Framework}\label{sec:fort_frame}

This section introduces the fortification framework. We define the notion of analytically fortified games and recall our main parallel repetition  and fortification theorems. We end by a discussion of the parameters of the resulting gap amplification results.
\subsection{Analytical fortification}\label{sec:anal_fort}
We distinguish between two variants of the notion of fortified games which we call \emph{weakly fortified games} and \emph{strongly fortified games}. Although the difference between the two may seem minor, this difference is in fact quite important in the quantum case. 
\begin{definition}[Fortified games]\label{def:fortification}
Let $\eps, \delta \in [0,1]$. A concatenated game $G'=(M\circ G\circ P)$ is called  weakly $(\eps,\delta)$- fortified against classical substragies if for any substrategy $f,g$ we have
\begin{equation}\label{eq:weakly_fortified_classical}
	\val(G',f,g)\leq (\val(G)+\eps)\cdot \Ex_{(x,y)\sim \mu} f(x)\, g(y) +\delta.
\end{equation}
Similarly, we define $G'$ to be weakly $(\eps,\delta)$-fortified against entangled substrategies if for any substrategy $\{A_{x'}^{a'}\}, \{B_{y'}^{b'}\}$ we have 
\begin{equation}\label{eq:weakly_fortified_quantum}
	\valst(G', \{A_{x'}^{a'}\},\{ B_{y'}^{b'}\})\leq (\valst(G)+\eps)\cdot \Ex_{(x,y)\sim \mu} \langle \psi| A_{x}\otimes B_y|\psi\rangle+\delta.
\end{equation}
If furthermore $\val(G)$ (resp.~$\valst(G)$) can be replaced by $\val(G')$, (resp.~$\valst(G')$) in the above then the game is called ``strongly fortified" against classical (resp.~quantum) substrategies.
\end{definition}
Note that our main results, Theorems \ref{thm:classical_fortification} and \ref{thm:main_fortification}, show how any game can be (weakly) fortified by concatenating it with good-enough spectral expanders.

Two remarks regarding the above definition are in order: 
\begin{itemize}
\item  Using~\eqref{eq:induced-val-lower}, we see that strong fortification implies weak fortification, as expected from the terminaology. 
\item From~\eqref{eq:induced-val-class} it follows that the two notions in fact coincide in the case of classical fortification, but this is no longer the case for quantum fortification.
\end{itemize}

Our notion of  fortified games  and that of~\cite{mosh2014,BSVV} are closely related. Essentially, in Definition~\ref{def:fortification} we have replaced the the condition for all $\delta$-large rectangles in~\eqref{eq:old_fortified} with a smoother condition. In terms of a precise relation, we can show the following. 
\begin{claim}\label{lem:easy}
Every $(\eps, \eps\delta)$ strongly fortified game is also $(2\eps,\delta)$ combinatorially fortified.
\end{claim}
\begin{proof} Consider a subgame given by $S \subseteq X, T \subseteq Y$ in $G$. To every strategy $(p,q)$ for $G_{S\times T}$, i.e., $p: S \to A$, $q: T \to B$, we can associate a natural substrategy $(f,g)$  by
\begin{equation}
f(x,a)= \begin{cases} 1 \qquad \text{if} \, \,   x\in S \wedge 	p(x)=a,\\ 
 0 \qquad \text{otherwise}	
 \end{cases}, \qquad
 g(y,b)= \begin{cases} 1 \qquad \text{if} \, \,   y\in T \wedge 	q(y)=b,\\ 
 0 \qquad \text{otherwise}	
 \end{cases}. 
\end{equation}
Then one can easily see
\begin{equation}\label{eq:subgame_vs_subst}
\val(G, f,g)= \val(G_{S\times T}, p,q) \cdot \mu(S\times T).\footnote{The term $\mu(S\times T)=\Ex_{(x,y)\sim \mu} f(x)g(y)$ is a natural scaling parameter playing an important role in our discussion as a measure of the  ``largeness" of a subgame or a substrategy. }
\end{equation}
Now assuming that rectangle $S\times T$ is $\delta$-large, i.e.~$\mu(S\times T)\geq \delta$, and since $G$ is fortified against classical substrategies, we have 
\begin{align}
\val(G_{S\times T}, p,q) &= \frac{\val(G,f,g)}{\mu(S \times T)} \\
						&\leq (\val(G) + \eps) \cdot \frac{\Ex_{(x,y)\sim \mu} f(x)g(y)}{\mu(S \times T)} + \frac{\delta \eps}{\mu(S \times T)} \\
						&\leq \val(G) + 2\eps
\end{align}
where in the second inequality we used  $\mu(S \times T) = \Ex_{(x,y)\sim \mu} f(x)g(y)$.
\end{proof}
We note that in Lemma \ref{lem:easy}, the reverse implication does not hold and the notion of analytically fortified  game is strictly stronger. In what follows, in the rare occasion when we call a game fortified (without specifying weak or strong) we mean strongly fortified.


\subsection{Parallel repetition of fortified games}
\label{sec:fortified-rep}
Using the definition of fortified games, it is straightforward to prove  the following parallel repetition theorem. 

\begin{theorem}[Basic parallel repetition]\label{thm:PR1}
	Let $G'_2$ be a  $(\eps,\delta)$-fortified game against classical substrategies. Then for any game $G'_1$ we have  
	\begin{equation}\label{eq:PR1-1}
	\val(G'_1\otimes G'_2)\leq (\val(G'_2)+\eps)\cdot \val(G'_1)	+ \delta\cdot |\Sigma_{G'_1}|,
\end{equation}
	where $\Sigma_{G'_1}$ is the total answer alphabet size (i.e.~the product of Alice and Bob's alphabets) of $G'_1$. 
\end{theorem}

We prove this theorem  in Section \ref{sec:concat} by adapting the proof of the analogous theorems in~\cite{mosh2014, BSVV} to the analytic setting. Unfortunately, while this theorem exemplifies the main idea behind our results, it is not directly useful for  applications. The reason for this is that the fortification procedure $G\rightarrow G'$ via concatenation  induces a  large blow-up in the alphabet size, $|\Sigma_{G'}|\approx |\Sigma_{G}|^D$, where $D=\frac{1}{\eps \sqrt{\delta}}$ is the degree of the expander graph chosen. As one iterates the repetition procedure $m$ times, the blow-up due to the additive term in~\eqref{eq:PR1-1} will be of order $\delta|\Sigma_{G'}|^{m-1} $. But typically $|\Sigma_{G'}|^{m-1}\gg |\Sigma_{G}|^{(m-1)/\sqrt{\delta}}$, leading to a term larger than $1$ and rendering the theorem useless.  

We resolve this problem by  proving an improved repetition theorem which exploits the fact that $G'$ takes the form of a concatenated game, whose  inner game $G$ has a much smaller alphabet. 

\begin{theorem}\label{thm:PR2_multiround}
Let $G'$ be a concatenated game, with inner game $G$, that is $(\eps,\delta)$-weakly fortified  against classical substrategies. If $\delta\cdot (m-1)\cdot |\Sigma_{G}|^{m-1}\leq \eta $ then 
\begin{equation}\label{eq:PR_multiround_classical}
\val( {G'}^{\otimes m} )\leq (\val(G)+\eps)^m + \eta.
\end{equation}
Similarly, if $G'$ is $(\eps,\delta)$ weakly-fortified against entangled substrategies and  $\delta\cdot (m-1)\cdot |\Sigma_{G}|^{m-1}\leq \eta $ then
\begin{equation}\label{eq:PR_multiround_quantum}
	\valst( {G'}^{\otimes m} )\leq (\valst(G)+\eps)^m + \eta.
\end{equation}

\end{theorem}

The main advantage of Theorem \ref{thm:PR2_multiround} compared to Theorem \ref{thm:PR1} is in the additive error, which is now is in terms $|\Sigma_G|$ rather than $|\Sigma_{G'}|$. What is important here is that the size of $|\Sigma_G|$ is independent of the fortification parameters $(\eps,\delta)$ whereas $|\Sigma_{G'}|$ grows exponentially as $\delta$ decreases. Let us also note that Theorem \ref{thm:PR2_multiround} is  quite general, and in particular applies to the multiplayer case.

\subsection{Gap amplification} Having stated our main parallel repetition, fortification, and biregularization theorems, all the main components of gap amplification are finally in place. Indeed, using $\val(G)=\val(G')$ Theorem~\ref{thm:PR2_multiround} implies our final gap amplification for the classical value. This matches the parameters of main results of \cite{mosh2014,BSVV} and extends it to more general settings.

Since quantumly we could have $\valst(G')< \valst(G)$, from~(\ref{eq:PR_multiround_quantum}) we cannot obtain
\begin{equation}\label{eq:PR_multiround_quantum_strong}
		\valst( {G'}^{\otimes m} )\leq (\valst(G')+\eps)^m + \eta.
\end{equation}
However, Theorem \ref{thm:main_fortification} and Theorem \ref{thm:PR2_multiround} are still sufficient to prove a gap amplification theorem for the case where the completeness holds against classical players and the soundness against the quantum ones.\footnote{E.g. as was the case in \cite{ito2012,vidick2013}.} To obtain a fully quantum gap amplification however, we need to appeal to the notion of \emph{ordered fortification} which, as we discussed, is a entangled-value preserving variant of the ordinary fortification. 
\begin{theorem*}[Theorem \ref{thm:fort_OF_main} restated]
Let $G$ be a game and $M$ and $P$ be two bipartite graphs as above. Let $G'_{OF}$ be constructed from $G$ and $G'=(M\circ G\circ P)$ as in Definition \ref{def:G_OF_intro}. Then, we have 
\[ \valst(G'_{OF})= \valst(G).\]
Furthermore if $M$ and $P$ are $\lambda$-spectral expanders and $\lambda\leq \frac{\eps^2\delta}{56}$, then $G'_{OF}$ is also $(\eps, \delta)$ weakly fortified.
\end{theorem*}

We stress that $G'_{OF}$ constructed above is itself a concatenated game with the inner game $G^{\oplus l}$, disjoint union of $l=\poly(\frac{1}{\eps^2\delta})$ copies of $G$. This means  the inner alphabet size of $G'_{OF}$ is precisely the same as $G$'s, and therefore there is fortunately no issue in terms of alphabet blow-up for applying Theorem \ref{thm:PR2_multiround} to $G'_{OF}$. So using $G'_{OF}$ instead of $G'$ in Theorem \ref{thm:PR2_multiround}, we can finally prove the analogue of (\ref{eq:PR_multiround_quantum_strong}) for $G'_{OF}$.

\paragraph{Parameters of gap amplification.}
We can now discuss the parameters of the gap amplification corollaries.  As in~\cite{mosh2014,BSVV}, the parameters are typically very good in terms of question sizes but much worse in terms of alphabet size. Here, we mostly focus our discussion to  gap amplification in the classical setting as the calculations in the quantum setting are similar. 

To understand the parameters, we need to only consider the soundness case. Suppose we are given a game $G$ with guarantee $\val(G)\leq 1-\tau$ and a target soundness value $\beta$. We choose $\eps= \tau/2$ and $m$ such that $(\val(G)+\eps)^m\leq \beta/2$. Hence, we have $m=\frac{\log(2/\beta)}{\log(1-\tau/2)}\leq \frac{2\log(2/\beta)}{\tau}$. We want 
\begin{equation}
\val({G'}^{\otimes m})\leq (\val(G)+\eps)	^m +\delta\cdot (m-1)|\Sigma_G|^{m-1} \leq \beta.
\end{equation}
 Hence, we just need to ensure  $\delta\cdot
|\Sigma_{G}|^{m-1}\leq \beta/2$. So we have $\delta=\frac{\beta}{(m-1)\cdot |\Sigma_{G}|^{m-1}}$. 

So what does the above mean in terms of the size of the final output of gap amplification ${G'}^{\otimes m}$. The question size is $|X|^m$ and $|Y|^m$ (since we have $|X'|=|X|$ and $|Y'|=|Y|$). Note that $m$ is essentially as small as we can hope for because even given a perfect parallel repetition theorem, we had to take $m\approx \frac{\log(1/\beta)}{\tau}$. Hence, the construction is essentially optimal in terms of question sizes. 

For the alphabet size, the situation is much worse. We have $|\Sigma_{G'}|=|\Sigma_G|^D$ where $D=O(\frac{\poly\log(1/\eps^2\delta)}{\eps^2\delta})$. This means (up to dominant factors) that $|\Sigma_{G'^{\otimes m}}| = |\Sigma_{G}|^{\frac{m^2|\Sigma_G|^{m-1}}{\beta}}$ which means that the alphabet is exponentially worse than basic parallel repetition which results in $|\Sigma_{G}|^{m}$. Note that however in typical settings where $|\Sigma_G|$ is constant and $\beta$ a small constant (or inverse logarithmic in size of $G$), this exponentially worse behavior of alphabet size does not cause a significant problem.

Next, let us consider the setting where  the completeness holds for classical players and soundness against  entangled players. In this case, we can just use Theorem \ref{thm:main_fortification} instead of Theorem \ref{thm:classical_fortification}, and hence all the calculations are precisely the same with $\eps$ and $\delta$ replaced with their squares.

Lastly, in the fully quantum case we need to use Theorem \ref{thm:fort_OF_main}. In this case, $m, \eps, \delta$ are chosen in precisely the same way. Alphabet size is also exactly the same as $G'_{OF}$ has the same alphabet size as $G'$. The only difference is that the question sizes in $G'_{OF}$ are slightly larger than $G'$: we have   $|X'|= |X|\cdot \poly(\frac{|\Sigma_G|^m}{\beta})$ and $|Y'|= |Y|\cdot \poly(\frac{|\Sigma_{G}|^m}{\beta})$. This is however arguably a minor blow-up since we typically expect that $|\Sigma_{G}|/\beta$ to be much smaller than $\textit{size}(G)=|X|\cdot |Y|$.

\section{Parallel Repetition Theorems}\label{sec:concat}
In this section we prove our main parallel repetition theorem. 
\begin{theorem*}[Theorem \ref{thm:PR2_multiround} restated]
	Let $G'$ be a concatenated game $(\eps,\delta)$-weakly fortified  against classical substrategies with inner game $G$. If $\delta\cdot (m-1)\cdot |\Sigma_{G}|^{m-1}\leq \eta $ then 
\begin{equation}\label{eq:classical_PR}
\val( {G'}^{\otimes m} )\leq (\val(G)+\eps)^m + \eta.
\end{equation}
Similarly, if $G'$ is $(\eps,\delta)$ weakly-fortified against entangled substrategies and  $\delta\cdot (m-1)\cdot |\Sigma_{G}|^{m-1}\leq \eta $ then
\begin{equation}\label{eq:quantum_PR}
	\valst( {G'}^{\otimes m} )\leq (\valst(G)+\eps)^m + \eta.
	\end{equation}
\end{theorem*}
The proof follows directly from the following proposition.
\begin{proposition}\label{prop:main_PR_prop}
Let $\{G'_i\}_{i=1}^t$ be a collection  of concatenated games with inner games $\{G_i\}_{i=1}^t$. Suppose that $G'_t$ is $(\eps, \delta)$ weakly fortified against classical substrategies. Then,
\begin{equation}\label{eq:sec_PR_1}
\val(G'_1\otimes G'_2\otimes \ldots \otimes G'_t)	\leq (\val(G_t)+\eps )\cdot \val(G'_1\otimes G'_2 \otimes \ldots \otimes  G'_{t-1})+ \delta \cdot \prod_{i=1}^{t-1} |\Sigma_{G_i}|.
\end{equation}
Similarly, if $G'_t$ is $(\eps, \delta)$ weakly fortified against quantum substrategies, then
\begin{equation}\label{eq:sec_PR_2}
\valst(G'_1\otimes G'_2\otimes \ldots \otimes G'_t)	\leq (\valst(G_t)+\eps )\cdot \valst(G'_1\otimes G'_2 \otimes \ldots \otimes  G'_{t-1})+ \delta \cdot \prod_{i=1}^{t-1} |\Sigma_{G_i}|.
\end{equation}
\end{proposition}
The key to proving Proposition~\ref{prop:main_PR_prop} is to work with the induced strategies. This allows us to get an additive error depending just on the alphabet size of the inner game. In the proof, we use the usual notation where a strategy missing an (answer) argument indicates summation over that variable. For example,
\[f(x_1, a_1, \ldots, x_{t-1},a_{t-1}, x_t)= \sum_{a_t} f(x_1, a_1, \ldots, x_{t-1},a_{t-1}, x_t, a_t).\]

\begin{proof}
We only prove \eqref{eq:sec_PR_1} as the proof of  \eqref{eq:sec_PR_2} follows the same structure. Also for simplicity, we focus on the case of two-player games as the proof of the multiplayer case is a straightforward extension.

Consider any strategies $f:X'_1\times A'_1 \times  \ldots \times  X'_t \times A'_t\rightarrow [0,1]$, $g:Y'_1\times B'_1 \times  \ldots \times  Y'_t \times B'_t\rightarrow [0,1]$. To clarify notation we will denote tuples $(z_1,\ldots,z_{t-1})$ as $\mathbf{z}_{<t}$. With this notation,  $\val(G'_1\otimes \ldots \otimes G'_t, f,g)$ is precisely 
\begin{equation}\label{eq:sec_PR_3}
\Ex_{(\mathbf{x}_{\leq t},\mathbf{y}_{\leq t})} \, \, \Ex_{\mathbf{x'}_{\leq t}}\Ex_{\mathbf{y'}_{\leq t} } \sum_{\mathbf{a'}_{\leq t},\mathbf{b'}_{\leq t}} \, \, \prod_{i=1}^t \, \,  V(a'_i(x_i), b'_i(y_i), x_i, y_i) \, f(\mathbf{x'}_{\leq t},\mathbf{a'}_{\leq t}) \cdot g(\mathbf{y'}_{\leq t},\mathbf{b'}_{\leq t}),
\end{equation}
where the expectations are according to $(x_i,y_i)\sim \mu_i$ and $x'_i\sim N(x_i)$ and $y'_i\sim N(y_i)$ for all $i=1,\ldots, t$. As usual let
\begin{equation}
	 f(\mathbf{x}_{< t},\mathbf{a}_{< t}, x'_t, a'_t)= \Ex_{\mathbf{x'}_{< t}\sim N(\mathbf{x}_{< t})} \, \sum_{a'_i(x_i)=a_i,\,i<t} f(\mathbf{x'}_{< t},\mathbf{a'}_{< t}, x'_t, a'_t).
\end{equation}
Using this notation, we can rewrite \eqref{eq:sec_PR_3} as 
\begin{equation}\label{eq:sec_PR_3b}
\Ex_{(\mathbf{x}_{< t}, \mathbf{y}_{< t})} \sum_{\mathbf{a}_{< t},\mathbf{b}_{< t}} \prod_{i=1}^{t-1} V(a_i,b_i,x_i,y_i)  \, S(\mathbf{x}_{< t},\mathbf{y}_{< t},\mathbf{a}_{< t},\mathbf{b}_{< t}),
\end{equation}
where $S(\mathbf{x}_{< t},\mathbf{y}_{< t},\mathbf{a}_{< t},\mathbf{b}_{< t})$ is given by 
\begin{equation}\label{eq:sec_PR_4}
\Ex_{(x_t,y_t)} \Ex_{x'_t}\Ex_{y'_t} \sum_{a'_t, b'_t} V(a'_t(x_t), b'_t(y_t),x_t,y_t) \,  f(\mathbf{x}_{< t},\mathbf{a}_{< t}, x'_t, a'_t)\cdot g(\mathbf{y}_{< t},\mathbf{b}_{< t}, y'_t, b'_t). 
\end{equation}
Consider the following substrategy $G'_t$: fix the first $2(t-1)$ arguments of $f$ to $(\mathbf{x}_{< t},\mathbf{a}_{< t})$ and the first $2(t-1)$ arguments of $g$ to $(\mathbf{y}_{< t},\mathbf{b}_{< t})$. Then~\eqref{eq:sec_PR_4} is precisely the value of this substrategy in $G'_t$. Since $G'_t$ is $(\eps, \delta)$ weakly fortified, it follows that 
\begin{equation}
 \eqref{eq:sec_PR_4} \leq (\val(G_t)+\eps)\cdot \Ex_{(x_t, y_t)} f(\mathbf{x}_{< t},\mathbf{a}_{< t}, x_t)\cdot g(\mathbf{y}_{< t},\mathbf{b}_{< t}, y_t) +\delta.
\end{equation}
Plugging this expression back into~\eqref{eq:sec_PR_3b}, $\val(G'_1\otimes \ldots \otimes G'_t, f,g)$ is bounded by 
\begin{align*}
(\val(G_t)+\eps) \Ex_{(\mathbf{x}_{\leq t},\mathbf{y}_{\leq t})} \sum_{\mathbf{a}_{< t},\mathbf{b}_{< t}} \prod_{i=1}^{t-1} V(a_i,b_i,x_i,y_i)  \,   f(\mathbf{x}_{< t},\mathbf{a}_{< t}, x_t)  \cdot    g(\mathbf{y}_{< t},\mathbf{b}_{< t}, y_t) 
+ \delta \cdot \prod_{i=1}^{t-1} |\Sigma_{G_i}|.
\end{align*}
To conclude we observe that
 \begin{equation} \label{eq:sec_PR_5}
  	\Ex_{(\mathbf{x}_{\leq t},\mathbf{y}_{\leq t})} \sum_{\mathbf{a}_{< t},\mathbf{b}_{< t}} \prod_{i=1}^{t-1} V(a_i,b_i,x_i,y_i)   \,   f(\mathbf{x}_{< t},\mathbf{a}_{< t}, x_t) \cdot  g(\mathbf{y}_{< t},\mathbf{b}_{< t}, y_t)
  \end{equation}
 is at most $\val(G'_1\otimes \ldots \otimes G'_{t-1})$, as for any fixed $(x_t, y_t)$ the functions $f(\cdot,x_t):X'_1\times A'_1\times \ldots \times X'_{t-1}\times A'_{t-1}\rightarrow [0,1]$ and $g(\cdot,y_t):Y'_1\times B'_1\times \ldots \times Y'_{t-1}\times B'_{t-1}\rightarrow [0,1]$ are valid strategies in $G'_1\otimes \ldots \otimes G'_{t-1}$.
\end{proof}

\begin{remark} 
Theorem \ref{thm:PR1} immediately follows from Proposition \ref{prop:main_PR_prop} by taking $t=2$ and considering the trivial concatenation $G'_1= G_1$, $G'_2= G_2$.  
\end{remark}

Theorem \ref{thm:PR2_multiround} follows easily.

\begin{proof}[Proof of Theorem \ref{thm:PR2_multiround}]
	We prove \eqref{eq:classical_PR} as  the proof of \eqref{eq:quantum_PR} is similar. 
	
The proof is by induction on $m$. The case $m=1$ is clear. By the induction hypothesis we have 
\[ \val({G'}^{\otimes (m-1)})\leq (\val(G)+\eps)^{m-1}+ \delta \cdot (m-2) |\Sigma_{G}|^{m-2}.\]
 Note that we can assume  $\val(G)+\eps<1$ otherwise \eqref{eq:classical_PR} holds trivially. Applying Proposition \ref{prop:main_PR_prop} we see that
 \begin{align*}
 	\val({G'}^{\otimes m}) &\leq (\val(G)+\eps)\cdot \val({G'}^{\otimes (m-1)})+ \delta\cdot |\Sigma_{G}|^{m-1} \\ 
 	&\leq (\val(G)+\eps)^{m}+ \delta\cdot (\val(G)+\eps)\cdot (m-2) |\Sigma_{G}|^{m-2} +  \delta\cdot |\Sigma_{G}|^{m-1} \\
 	&\leq (\val(G)+\eps)^{m} +\delta \cdot (m-1) \cdot |\Sigma_G|^{m-1}.
 \end{align*}
\end{proof}


\section{Classical Fortification}\label{sec:classical}

In this section we prove our main theorem regarding the fortification of  classical games. Beside providing a short and self-contained treatment of the main result of \cite{mosh2014,BSVV}, it serves as preparation for the analysis of Section \ref{sec:quantum_generalization}.

\begin{theorem*}[Theorem \ref{thm:classical_fortification} restated]
Let $G$ be a biregular game, $M$ and $P$ two bipartite $\lambda$-spectral expanders. If $\lambda\leq \frac{\eps}{2}\sqrt{\frac{\delta}{2}}$, then  the concatenated game $G'=(M\circ G\circ P)$ is $(\eps,\delta)$ weakly fortified against classical substrategies.
\end{theorem*}
We note that it follows from~\cite[Appendix C]{BSVV} that the dependence $\lambda$ and $\delta$ in Theorem \ref{thm:classical_fortification} is up to constant factors optimal. On the other hand, the tightness of dependence of $\eps$ and $\delta$ does not seem to follow from \cite{BSVV} lower bound (however, $\delta$ is by far the more significant of the two parameters).

\subsection{Proof of Theorem \ref{thm:classical_fortification}}

We start with a simple claim whose proof we will defer to the end of the subsection.
\begin{claim}\label{cor:exp-averages}
Let  $M = (X' \times X, E)$ and $N=(Y'\times Y,F)$ be two biregular bipartite graphs that are $\lambda$-spectral expanders. Let $\mu$ be a distribution on $X\times Y$ such that both marginals of $\mu$ are uniform. Let $f:X'\to\R$ and $g:Y'\to R$ be any functions, and denote $f:X\to \R$ and $g:Y\to \R$ the functions $f:x\mapsto\Ex_{x'\sim N(x)}f(x)$, $g:y\mapsto \Ex_{y'\sim N(y)} g(y)$ respectively. Then
$$ \Ex_{(x_1,y_1)\sim \mu} \Big| f(x_1)g(y_1) - \Ex_{(x_2,y_2)\sim \mu} f(x_2)g(y_2)\Big| \,\leq\,2\sqrt{2}\lambda\Big(\Ex_{x'\sim X'} |f(x')|^2 \Big)^{1/2}\Big(\Ex_{y'\sim Y'} |g(y')|^2 \Big)^{1/2}$$ 
and
$$  \Big| \Ex_{x_1\sim X}f(x)\Ex_{y_1\sim Y} g(y_1) - \Ex_{(x_2,y_2)\sim \mu} f(x_2)g(y_2)\Big| \,\leq\,2\lambda^2\Big(\Ex_{x'\sim X'} |f(x')|^2 \Big)^{1/2}\Big(\Ex_{y'\sim Y'} |g(y')|^2 \Big)^{1/2}.$$ 
\end{claim}

We prove a slightly stronger statement which implies  Theorem \ref{thm:classical_fortification}. Let $f,g$ be any substrategies for $G$, and let $\gamma = \Ex_{(x,y)\sim \mu} f(x)g(y)$. We claim that 
\begin{equation}\label{eq:thm11-1}
\val(G',f,g) \leq  \val(G) \gamma+ 2\sqrt{2} \,\lambda \sqrt{\gamma} + 4\,\lambda^2.
\end{equation}
To deduce the bound claimed in Theorem~\ref{thm:classical_fortification} from~\eqref{eq:thm11-1} we distinguish two cases. Either $\gamma \leq \delta$, in which case using the trivial estimate $\val(G',f,g) \leq \gamma$ the bound immediately follows. Or $\gamma > \delta$, in which case 
\begin{align*}
 \val(G) \gamma+ 2\sqrt{2} \lambda \sqrt{\gamma} + 4\lambda^2 &\leq \gamma(\val(G) + 2\sqrt{2} \lambda \delta^{-1/2}) + 4 \lambda^2 \\
&\leq \gamma(\val(G) + \eps) + \delta
\end{align*}
given the relation between $\eps,\delta$ and $\lambda$ expressed in the theorem.

It remains to prove~\eqref{eq:thm11-1}. Fix substrategies $f$ and $g$. We have
\begin{align*}
\val(G',f,g)&= 	\Ex_{(x,y)\sim \mu} \sum_{V(a,b,x,y)=1} f(x,a)\cdot g(y,b)\\
&= \Ex_{(x,y)\sim \mu} f(x)g(y)  \sum_{V(a,b,x,y)=1} \frac{f(x,a)}{f(x)} \cdot \frac{g(y,b)}{g(y)},
\end{align*}
where we adopt the convention that $0/0 = 0$. Using the triangle inequality, 
\begin{align}
\val(G',f,g)&\leq 	
 \gamma \Ex_{(x,y)\sim \mu}  \sum_{V(a,b,x,y)=1} \frac{f(x,a)}{f(x)} \cdot \frac{g(y,b)}{g(y)} + \Ex_{(x,y)\sim \mu} \big|f(x)g(y)  - \gamma\big|\notag\\
&\leq \gamma\val(G) +  \Ex_{(x,y)\sim \mu} \big|f(x)g(y)  - \gamma\big|,\label{eq:thm11-2}
\end{align}
where the second inequality follows since $(x,a)\mapsto f(x,a)/f(x)$ and $(y,b)\mapsto g(y,b)/g(y)$ form a valid pair of strategies for $G$. It remains to estimate  the second term above.  
Applying the first inequality in Claim~\ref{cor:exp-averages}, 
\begin{align*}
 \Ex_{(x,y)\sim \mu} \big|f(x)g(y)  - \gamma\big| &\leq 2\sqrt{2} \lambda \Big( \Ex_{x'\sim X'} |f(x')|^2\Big)^{1/2}\Big( \Ex_{y'\sim Y'} |g(y')|^2\Big)^{1/2}\\
&\leq 2\sqrt{2} \lambda \Big( \Ex_{x'\sim X'} f(x') \,\Ex_{y'\sim Y'} g(y')\Big)^{1/2}\\
&\leq 2\sqrt{2} \lambda \sqrt{\gamma+2\lambda^2}\\
&\leq 2\sqrt{2} \lambda (\sqrt{\gamma}+ \sqrt{2} \lambda) \\
&= 2\sqrt{2} \lambda \sqrt{\gamma} + 4\lambda^2, 
\end{align*}
where in the second inequality we used $0\leq f(x'),g(y')\leq 1$ for all $x',y'$ and the third uses the second inequality in Claim~\ref{cor:exp-averages}. Together with~\eqref{eq:thm11-2} this proves~\eqref{eq:thm11-1}.

Finally, we prove Claim \ref{cor:exp-averages}.
\begin{proof}[Proof of Claim \ref{cor:exp-averages}]
For the first inequality, write
\begin{align*}
 \Ex_{(x_1,y_1)\sim \mu} &\big| f(x_1)g(y_1) - \Ex_{(x_2,y_2)\sim \mu} f(x_2)g(y_2)\big|\\
 &\leq  \Ex_{(x_1,y_1),(x_2,y_2)\sim \mu} \big(| f(x_1)-f(x_2)||g(y_1)| + |f(x_2)||g(y_1)-g(y_2)|\big)\\
&\leq \Big( \Ex_{x_1,x_2\sim X} |f(x_1)-f(x_2)|^2 \Big)^{1/2} \Big(\Ex_{y_1\sim Y} |g(y_1)|^2 \Big)^{1/2} \\
&\hskip3cm+\Big(\Ex_{x_2\sim X} |f(x_2)|^2 \Big)^{1/2}\Big( \Ex_{y_1,y_2\sim Y} |g(y_1)-g(y_2)|^2 \Big)^{1/2} \\
&\leq \lambda \Big( \Ex_{x'_1,x'_2\sim X'} |f(x'_1)-f(x'_2)|^2 \Big)^{1/2}   \Big(\Ex_{y'_1\sim Y'} |g(y'_1)|^2 \Big)^{1/2} \\
&\hskip3cm+\lambda\Big(\Ex_{x'_2\sim X'} |f(x'_2)|^2 \Big)^{1/2}\Big( \Ex_{y'_1,y'_2\sim Y'} |g(y'_1)-g(y'_2)|^2 \Big)^{1/2},
\end{align*}
where the last inequality uses Proposition~\ref{prop:expander}.
Now note that $ \Ex_{x'_1,x'_2\sim X'} |f(x'_1)-f(x'_2)|^2\leq 2\Ex_{x'\sim X'} |f(x')|^2$. Applying a similar bound for $g$ gives us the first inequality. 
For the second, write 
\begin{align*}
\big| \Ex_{(x_2,y_2)\sim \mu} f(x_2)g(y_2)  &\Ex_{x_1\sim X}f(x)\Ex_{y_1\sim Y} g(y_1)\big| 
\\
&= \big| \Ex_{(x_2,y_2)\sim\mu, x_1\sim X,y_1\sim Y} (f(x_1)-f(x_2))(g(y_1)-g(y_2)) \big|\\
&\leq \Big( \Ex_{x_1,x_2\sim X} ( f(x_1)-f(x_2))^2 \Big)^{1/2}\Big( \Ex_{y_1,y_2\sim Y} ( g(y_1)-g(y_2))^2 \Big)^{1/2}\\
&\leq \lambda^2 \Big( \Ex_{x'_1,x'_2\sim X'} ( f(x_1)-f(x_2))^2 \Big)^{1/2}\Big( \Ex_{y'_1,y'_2\sim Y'} ( g(y'_1)-g(y'_2))^2 \Big)^{1/2}\\
&\leq 2\lambda^2 \Big(\Ex_{x'\sim X'} |f(x')|^2 \Big)^{1/2}\Big(\Ex_{y'\sim Y'} |g(y')|^2 \Big)^{1/2}.
\end{align*}
\end{proof}
\subsection{A simple multiplayer fortification}\label{sec:multiplaer_classical}
 The following is a simple fortification theorem for $k$-player games. Since  Theorem \ref{thm:PR2_multiround} applies equally well to the multiplayer setting, we get a hardness amplification result for classical multiplayer games.

\begin{theorem}\label{thm:multiplyaer_classical}
Let $G$ be a $k$-player game. Suppose $G'$ is given by composing each of the $k$ sides of $G$ by a $\lambda$-spectral expander where $\lambda\leq 2\delta/k$. Then $G'$ is a $(0, \delta)$ fortified game.\footnote{Although there is no $\eps$ dependence in the above, when applied to 2-player games the theorem is still weaker than Theorem \ref{thm:classical_fortification} because of the worse dependence on $\delta$ -- which is the more crucial parameter than $\eps$. }
\end{theorem}
\begin{proof}
Consider a classical substrategy for $G'$ given by $f_i:X'_i \times A'_i\rightarrow \R^+$ for $i=1,2,\ldots,k$.  As usual, denote $f_i:X_i\times A_i\rightarrow \R^+$ the projection of $f_i$ to the inner game $G$. By definition,
\[ \val(G, \{f_i\}_{i=1}^k)=\Ex_{(x_1, \ldots, x_k)} \sum_{a_1, a_2, \ldots, a_k} V(a_1, \ldots, a_k, x_1, \ldots, x_k)\cdot f_1(x_1, a_1)\cdot f_x(x_2, a_2)\ldots f_k(x_k, a_k).\]
We can rewrite the above as
\[ \Ex_{(x_1, \ldots, x_k)} \prod_{i=1}^k f_i(x_i) \sum_{a_1,\ldots, a_k} V(a_1, \ldots, a_k, x_1, \ldots, x_k)\cdot \frac{f_1(x_1, a_1)\cdot f_x(x_2, a_2)\ldots f_k(x_k, a_k)}{f(x_1)\cdot f(x_2)\ldots \cdot f(x_k)}.\]

Let $\gamma=  \Ex_{(x_1, \ldots, x_k)} \prod_{i=1}^k f_i(x_i)$. Applying the triangle inequality,
\[ \val(G, \{f_i\}_{i=1}^k)\leq \gamma \cdot \val(G) + \Ex_{(x_1, \ldots, x_k)} |\prod_{i=1}^k f_i(x_i)-\gamma|.\]
To conclude it will suffice to show the second term above is at most $\delta$. Let $\overline{f_i} = \Ex_{x_i} f(x_i)$. Then
\begin{align*}
	 \Ex_{(x_1, \ldots, x_k)} \left|\prod_{i=1}^k f_i(x_i)-\gamma \right | &\leq  \Ex_{x_1, \ldots, x_k} \left|\prod_{i=1}^k f_i(x_i)-\prod_{i=1}^k \overline{f_i}\right| + \Ex_{(x_1, \ldots, x_k)} \left|\prod_{i=1}^k \overline{f_i} -\gamma \right|  \\
	  &= \Ex_{(x_1, \ldots, x_k)} \left|\prod_{i=1}^k f_i(x_i)-\prod_{i=1}^k \overline{f_i}\right | + \left  | \prod_{i=1}^k \overline{f_i} - \Ex_{x_1, \ldots, x_k} \prod_{i=1}^k f_i(x_i) \right|
	 \\ &\leq 2 \cdot \Ex_{(x_1, \ldots, x_k)}  \left |\prod_{i=1}^k f_i(x_i)-\prod_{i=1}^k \overline{f_i} \right | \\
	 & \leq 2\,\sum_{i=1}^k \Ex_{x_i} |f_i(x_i)- \overline{f_i}|,
\end{align*}
where the first equality is by definition of $\gamma$, the second inequality by convexity of $|\cdot|$, and the last follows from 
\[ \ |f_1(x_1)f_2(x_2)\ldots f_k(x_k)- \overline{f_1} \overline{f_2}\ldots \overline{f_k}|\leq \sum_{\ell=1}^k \left| \prod_{i=1}^{\ell-1}f_i(x_i) \cdot \prod_{i=\ell}^k  \overline{f_i} - \prod_{i=1}^{\ell} f_i(x_i) \cdot \prod_{i=\ell+1}^k \overline{f_i}\right | \leq \sum_{i=1}^k \Ex_{x_i} |f_i(x_i)- \overline{f_i}|. \]
Hence,
\[ \Ex_{x_1,\ldots,x_k} \left|\prod_{i=1}^k f_i(x_i)-\gamma \right| \leq 2 \sum_{i=1}^k \left( \Ex_{x_i} ( f_i(x_i)-\overline{f_i})^2 \right)^{1/2} \leq 2\lambda  \sum_{i=1}^k  \left(\Ex_{x'_i} ( f_i(x'_i)-\overline{f_i})^2\right)^{1/2} \leq 2\lambda k.\]
The desired result follows. 
\end{proof}


\section{Reducing Strong to Weak Fortification for  Entangled Games}\label{sec:modified_fortification}
In this section, we start working toward the problem of fortifying games in the entangled case. In particular, we show how Theorem \ref{thm:fort_OF_main} follows from Theorem \ref{thm:main_fortification}. Let $G=(X\times Y, A\times B, \mu,V)$ be a two-player game. 
\begin{definition}\label{def:disjoint_copy}
For a game $G$ and integer $l\in \N$ let $G^{\oplus l}$ denote the disjoint union of $l$ copies of $G$.	
\end{definition}
Suppose that $M$ and $P$ are regular bipartite graphs over $X'\times X$ and $Y'\times Y$, respectively. Suppose further that $M$ and $P$ are balanced, i.e.~$|X'|=|X|$ and $|Y'|=|Y|$. Let $d_M$ and $d_N$ denote the degree of vertices $M$ and $P$, respectively. (Note that since the graphs are balanced and regular, the left and right degrees are the same.) 

Following \cite{BSVV}, we assume that $M$ and $P$ are explicit bipartite almost-Ramanujan expanders, as provided e.g.~by \cite{bilu2006}, for which the second-largest singular values $\lambda_M$ and $\lambda_P$ of $A_M$ and $A_P$ (the normalized adjacency matrices) respectively satisfy
\begin{equation}
\lambda_M= O \left (\frac{\poly(\log d_M) }{\sqrt{d_M}} \right), \qquad \lambda_{P}= O\left (\frac{\poly(\log d_P)}{\sqrt{d_P}} \right).
\end{equation}
Note that  if $d_M, d_P=\wtilde{\Omega}(\frac{1}{\eps^2\delta})$ then Theorem \ref{thm:main_fortification} implies that $G'=(M\circ G\circ P)$  is $(\eps,\delta)$ weakly-fortified. Next we recall the definition of $G'_{OF}$ from the introduction.

\paragraph{Ordered fortification.} Let $G$, $M$, $P$ and $G'=(M\circ G\circ P)$ be as above. let $l= \max \{ d_M,d_P \}$.  In $G'_{OF_l}$ (or simply $G'_{OF}$ where $l=\max\{d_M, d_P\}$) the referee samples questions $(x,y)$ as in $G$ and selects two random neighbors $x'\in X'$ and $y'\in Y'$ of $x$ and $y$ in $M$ and $P$ respectively. Then the referee selects two random injective maps $r_{x'}:N(x')\rightarrow [l]$ and $s_{y'}:N(y')\rightarrow [l]$ under the condition $r_{x'}(x)= s_{y'}(y)$. Alice's question then is the pair $(x', r_{x'})$ and Bob's is the pair $(y', s_{y'})$. Alice outputs an answer tuple $a':N(x')\rightarrow A$ and Bob  $b':N(y')\rightarrow B$. The players win if $V(a'(x), b'(y), x,y)=1$. 

\begin{remark}
Note that $G'_{OF}$ has exactly the same answer alphabet size as $G'$, the question sizes $|X'_{OF}|$ and $|Y'_{OF}|$ are larger than in $G'$. This blow-up can be mitigated as follows. It turns out that in Definition \ref{def:modified_fortified} the use of the complete set $S_{(d, l)}$ is unnecessary. More precisely, from the proof of the main claim of this section, Claim~\ref{claim:SVs_modified_fortified} below, it will be clear that the only condition required is that the permutations be chosen from a pairwise independent subset of $S_{(d,l)}$. Selecting the smallest possible such subset lets us reduce the blow-up in the size of the question sets from a multiplicative $D!$ down to $\poly(D)=\poly(\frac{1}{\eps^2 \delta})$. We omit the details.
\end{remark}


Although it may not be immediately apparent, it is possible to view $G'_{OF}$ as a concatenated game.  Let $G^{\oplus l}$ be as in Definition \ref{def:disjoint_copy}. Note that $G^{\oplus l}$ has exactly the same classical and entangled value as $G$. Let $S_{(d_M, l)}$  denote the set of all injective maps from $[d_M]\rightarrow [l]$. Fix maps $u_{x'}: N(x')\rightarrow [d_M]$ and $v_{y'}:N(y')\rightarrow [d_M]$ ordering the neighborhoods of each $x',y'$ in an arbitrary way.


\begin{definition}\label{def:modified_fortified}
Let $M$ be a regular bipartite graph over $X'\times X$ as above	. We define $\wtilde{M}$ as a bipartite graph over $X'_{OF}:=X'\times S_{(d_M, l)}$ and $X_{OF}:=X\times [l]$ where 
\[ (x',\pi)\sim_{\tilde{M}} (x, i) \qquad \Longleftrightarrow \qquad  \pi(u_{x'}(x))=i.\]
We define $\wtilde{P}$ from $P$ in a similar way.
\end{definition}
Note that here $\pi\circ u_{x'}$ exactly corresponds to $r_{x'}:N(x')\rightarrow [l]$ map from the original definition of $G'_{OF}$. Hence, we obtain the following alternative characterization of $G'_{OF}$. 
\begin{proposition}\label{prop:modified_frotified_concat}
The game $G'_{OF}$ constructed above is a concatenated game given by  
\[ G'_{OF}= (\wtilde{M} \circ G^{\oplus l} \circ \wtilde{P}).\]	
\end{proposition}
Next, we show that ordered fortification preserves the entangled value (the classical value is also preserved but that is not important here). 
\begin{proposition}\label{prop:value_preserve_entangled}
	We have $\valst(G'_{OF})= \valst(G)$.
\end{proposition}
\begin{proof}
In one direction we have $\valst(G'_{OF})\leq \valst(G^{\oplus l})= \valst(G)$ where we used Propositions \ref{prop:val_inequality_classical_quantum} and \ref{prop:modified_frotified_concat}. For the other direction, consider any entangled strategy $(\ket{\psi},\{A_x^a\},\{B_y^b\})$  for $G$. We construct a strategy for $G'_{\oplus}$ that achieves the same value. The provers share $l$ copies of the state $\ket{\psi}$, and each copy is assigned a unique label $i \in [l]$. Alice and Bob receive questions $(x',r_{x'})$ and $(y',s_{y'})$, respectively. For each $x\in N(x')$, Alice applies $\{A_{x}^a\}$ to the $r_{x'}(x)$-th copy of $|\psi\rangle$. Bob applies a similar strategy.
 
  Since by construction the ``true questions'' $x^*$ and $y^*$ are given the same label, the distribution of answers obtained for $x^*$ and $y^*$ is identical to the distribution of answers obtained while playing $G$ using $(\ket{\psi},\{A_x^a\},\{B_y^b\})$, hence achieving the same winning probability.
\end{proof}


The main technical step in reducing Theorem \ref{thm:fort_OF_main} to Theorem \ref{thm:main_fortification} is an analysis of the singular values of $\wtilde{M}, \wtilde{N}$ in terms of the singular values of $M$ and $N$. We prove the following.

\begin{claim}\label{claim:SVs_modified_fortified}
Let $M$ be a bipartite graph over $X' \times X$ as above and let $\lambda_M$ denote the second largest singular value of $M$.  Let $\tilde{M}$ be as in Definition \ref{def:modified_fortified}. Then, 

\[ \lambda_{\tilde{M}} \leq \max\left\{\lambda_M, \frac{1}{\sqrt{d_M-1}}\right\}.\]

%
\end{claim}

Since in our case $\lambda_M=O\left(\poly(\log d_M)/d_M\right)$, Claim~\ref{claim:SVs_modified_fortified} implies that $\lambda_{\tilde{M}}$ satisfies the same bound. Also note that a similar statement of course applies to  $\tilde{P}$ and $\lambda_{\tilde{P}}$.  So we see that Theorem \ref{thm:main_fortification}, Propositions \ref{prop:value_preserve_entangled} and \ref{prop:modified_frotified_concat}, and Claim~\ref{claim:SVs_modified_fortified} together imply Theorem \ref{thm:fort_OF_main}; it remains to prove the latter.

\begin{proof}[Proof of Claim \ref{claim:SVs_modified_fortified}]
Recall that by assumption $M$ is a  regular balanced bipartite graph. Let $d:=d_M$ the degree of vertices in $M$. 
The normalized adjacency matrix of $\tilde{M}$ is given by 
\begin{equation}
A_{\tilde{M}}((x',\pi) ,(x,i)) = \begin{cases} \frac{1}{d}\cdot\sqrt{\frac{(l-d)!}{ (l-1)!}} \qquad  & (x', \pi) \sim_{\tilde{M}} (x,i)\\
 0 \qquad  &  (x', \pi)\not\sim_{\tilde{M}}(x,i)
 \end{cases}.
\end{equation}
We relate the second largest singular value $\lambda_M$ of $A_M$ and the second largest singular value $\lambda_{\tilde{M}}$ of $\tilde{M}$ by relating the eigenvalues of $B= A^\top_{M}A_M$ and $C= A^\top_{\tilde{M}}A_{\tilde{M}}$. 
We can explicitly compute the entries of $B$ and $C$. For $B$,
\begin{equation}  B(x_1, x_2)= \frac{ |\{x'\in X'\, : \, \{x_1, x_2\}\subset N(x')\}|}{ d^2},\end{equation}
and in particular $B(x,x)=\frac{1}{d}$ for all $x\in X$. 
To compute entries of $C$, first note that when $x_1\neq x_2$ and $i\neq j$ we have
\begin{equation}
C((x_1, i), (x_2, j))= \frac{| \{x' \in X' \, : \{x_1, x_2\}\subset N(x')\}|\cdot (l-d)! }{d^2 (l-1)!}\cdot \frac{(l-2)!}{(l-d)!}= \frac{B(x_1, x_2)	}{l-1}. 
\end{equation}
Finally, observe the following special cases:
\begin{itemize} 
\item $C((x,i), (x,i))=\frac{1}{d}$. 
\item $C((x_1, i), (x_2, i))=0$ if $x_1\neq x_2$. 
\item $C((x,i), (x,j))=0$ when $i\neq j$. 
	\end{itemize}
	
	Let $\wtilde{B}= B- \frac{1}{d} \Id$ and $\wtilde{C}= C-\frac{1}{d} \Id$. Let $\wtilde{J}= \frac{1}{l-1} (J - \Id)$ be the $l \times l$ matrix that is $(l-1)^{-1}$ in the off-diagonal entries, and $0$ along the diagonal. Then it is easy to see that
	\begin{equation}
 \wtilde{C}= \wtilde{B}\otimes \wtilde{J}.
	\end{equation}
The matrix $\wtilde{J}$ has a single eigenvalue equal to $1$ and $l-1$ eigenvalues equal to $-\frac{1}{l-1}$, and $\wtilde{B}$ has a single eigenvalue equal to $1-1/d$ and the remaining are in the range $[-\frac{1}{d}, \lambda_M^2-\frac{1}{d}]$. It follows that the top eigenvalue of $C=\wtilde{C}+\frac{1}{d} \Id$ is $1$ (as expected) and the next one satisfies
	\[ \lambda_{\tilde{M}} \leq \max\left \{ \lambda_M, \sqrt{\frac{1}{d(l-1)} +\frac{1}{d}} \right\},\]
	which is bounded by  $ \max\left\{\lambda_M, \frac{1}{\sqrt{d-1}}\right\}$ since
 $l\geq d$. 
\end{proof}

\section{Weak Fortification of Entangled Games}\label{sec:quantum_generalization}
 Our goal in this section is to prove  the following.
	\begin{theorem*}[Theorem \ref{thm:main_fortification} restated]
 	Let $G'=(M\circ G\circ P)$ be a concatenated game obtained by concatenating two sides of a game $G$ with some $\lambda$-spectral expanders $M$ and $P$. If $\lambda\leq \frac{\eps^2\delta}{56}$, then  $G'$ is $(\eps,\delta)$ weakly-fortified against entangled substrategies. 
 \end{theorem*}
We need some basic matrix analytic facts.

 \subsection{Basic lemmas}
\paragraph{Choi-Jamiolkowski isomorphism.} We make use of the correspondence between bipartite states $|\psi\rangle\in \mathcal H_1\otimes \mathcal H_2$ and linear operators $L:\mathcal H_2^*\rightarrow \mathcal H_1$ given by the Choi-Jamiolkowski isomorphism. Explicitly, let $|\psi\rangle \in \C^d\otimes \C^d$ be a quantum state and consider a Schmidt basis for $|\psi\rangle$ so we have $|\psi\rangle= \sum_{i=1}^d \sqrt{\lambda_i} |i\rangle |i\rangle$ where $\lambda_i\in \R^{\geq 0}$, up to a local change of basis. Set 
\begin{equation}
\rho := \sum_{i=1}^{d} \lambda_i |i\rangle\langle i|.
\end{equation}
\begin{proposition}\label{prop:choi}
Let $Z,W$ be two linear operators acting on $\C^d$ and let $|\psi\rangle$ and $\rho$ be as above. Then,
\[
	\langle \psi|Z\otimes W|\psi\rangle= \Tr(Z\rhohalf W^T \rhohalf).  
\]	
\end{proposition}
\begin{proof}
Both expressions evaluate to $\sum_{i,j=1}^{d} \sqrt{\lambda_i \lambda_j} \, Z_{ij}\cdot W_{ij}$. 
\end{proof}
For a density matrix $\rho$ and a matrix $A$ for convenience we sometimes denote $\Tr(A\rho)$  by $\Tr_\rho(A)$. 
\paragraph{Matrix norms and inequalities.} 
The Frobenius norm of a matrix $A\in \C^{n\times m}$ is defined as $\|A\|_F = \sqrt{\Tr(AA^\dagger)}$. The trace norm is defined as $\|A\|_{tr} = \Tr\sqrt{AA^\dagger}$. 
The following analogue of Proposition~\ref{prop:expander} will be used repeatedly in our argument.

\begin{claim}\label{claim:expansion}
Let $M$ be a bipartite $\lambda$-spectral expander on vertex set $X'\cup X$. Let $\{A_{x'}\}_{x'\in X'}$ and $\rho$ be positive semidefinite matrices. For all $x \in X$, define $A_x=\Ex_{x' \sim N(x)} A_{x'}$ and define $A = \Ex_{x \sim \mu} A_x$. Then  
\begin{equation}\label{eq:quantum_smoothing}
\Ex_{x\sim \mu } \Tr_\rho((A_x-A)^2)\leq 2 \lambda^2\cdot \Ex_{x'\sim  \mu'} \Tr_\rho(A_{x'}^2).
\end{equation}
\end{claim}
\begin{proof}
Define $S_{x'}= \rhohalf A_{x'}$, $S_x=\rhohalf A_{x}$,  and $S=\Ex_{x}S_x= \rhohalf A$.
Using that $M$ is a bipartite $\lambda$-spectral expander, for any fixed entry $(i,j)$
\begin{equation}
\Ex_{x} |(S_x)_{ij}- S_{ij}|^2\leq \lambda^2\cdot \Ex_{x'} |(S_{x'})_{ij}- S_{ij}|^2\leq 2 \lambda^2\cdot \Ex_{x'} |(S_{x'})_{ij}|^2
\end{equation}
Summing over all entries,
\begin{equation}
	\Ex_{x} \sum_{i,j}  |(S_x)_{ij}- S_{ij}|^2 = \Ex_{x} \| S_x -S \|_F^2 \leq 2 \lambda^2 \Ex_{x'} \sum_{i,j} |(S_{x'})_{ij}|^2 = 2 \lambda^2 \Ex_{x'} \| S_{x'} \|_F^2.
\end{equation}
Observing that $\Tr_\rho((A_x-A)^2) = \|S_x-S\|_{F}^2$ and $\| S_{x'} \|_F^2 = \Tr_\rho(A_{x'}^2)$, we obtain the desired result.
\end{proof}
If $A$ has singular value decomposition $A=UJ V^\dagger$ its pseudo-inverse is $A^{-1}=VJ^{-1}U^\dagger$, where $J^{-1}$ is obtained from $J$ by taking the reciprocal of non-zero diagonal entries. A simple consequence of the singular value decomposition is  the following:
 
\begin{fact}\label{fact:PSDness}
 Let $A$ be an $n\times n$ matrix. Then there exists a unitary matrix $\mathcal U$ such that $\mathcal U A$ is positive semi-definite.
 \end{fact}
 \begin{proof}
 Write the SVD as $A=U J V^\dagger$, and choose $\mathcal U= VU^\dagger$.
 \end{proof}
We make frequent use of the matrix Cauchy-Schwarz inequality. 
\begin{proposition}\label{fact:cor_cs}
For any two matrices $S,T$ we have 
\[ \Tr(ST^\dagger)\leq \Tr(SS^\dagger)^{1/2}\cdot \Tr(TT^\dagger)^{1/2}= \|S\|_F\|T\|_F.	\]
If $S$ and $T$ are Hermitian, 
\[ \Tr(STST)\leq \Tr(S^2T^2).\]	
\end{proposition}
Finally, we need a variant of Powers-St{\o}rmer inequality due to Kittaneh~\cite{kit86}. This also played a role in the analysis of~\cite{DinurSV14}.
\begin{lemma}[\cite{kit86}]\label{lemma:kit}
	Let  $X,Y$ be positive semidefinite matrices. Then
	\[ \Tr\left ( (X-Y)^4\right) \leq \Tr\left( (X^2- Y^2)^2\right). \]
\end{lemma}
\subsection{Proof of Theorem \ref{thm:main_fortification}}
At a high level, the proof of Theorem~\ref{thm:main_fortification} follows the same outline as the classical proof of Section~\ref{sec:classical}. 

Consider a substrategy $\{A_{x'}^{a'}\}_{(x',a')\in X'\times A'}$, $\{B_{y'}^{b'}\}_{(y',b')\in Y'\times B'}$ for $G'$. Define $A_x = \Ex_{x \sim N(x')} A_{x'}$ and $B_y = \Ex_{y  \sim N(y')} B_{y'}$.\footnote{In what follows, we assume without loss of generality that all $A_{x'}$ and $B_{y'}$ are invertible. Note that proving Theorem \ref{thm:main_fortification} for this subset of substrategies suffices. This follows by a limiting argument because of the continuity  of \eqref{eq:weakly_fortified_quantum} in $A_{x'}$ and $B_{y'}$.} Define $A = \Ex_{x \sim  \mu_X} A_x$ and $B = \Ex_{y \sim \mu_Y} B_y$.
To prove Theorem \ref{thm:main_fortification} we must analyze the following expression:
\begin{align*} 
\valst(G',\{A_{x'}^{a'}\},\{B_{y'}^{b'}\}) &= \Ex_{(x,y)\sim \mu} \Ex_{x' \sim N(x), y'\sim  N(y)} \sum_{a',b'}V(a'_{x'}(x),b'_{y'}(y),x,y)\cdot  \Tr(A_{x'}^{a'}\rhohalf B_{y'}^{b'}\rhohalf) \\
&= \Ex_{(x,y)\sim \mu} \sum_{a,b} V(a,b,x,y) \cdot \Tr(A_x^a \rhohalf B_y^b \rhohalf) \\
&= \Ex_{(x,y)\sim \mu } \Tr(A_x\rhohalf B_y\rhohalf) \cdot \sum_{V(a,b,x,y)=1}\, \frac{\Tr(A_x^{a}\rhohalf B_y^b\rhohalf)}{\Tr(A_x\rhohalf B_y\rhohalf )},
\end{align*}
where $A_x^a$ and $B_y^b$ are defined as in~\eqref{eq:projected_operators}, and we use the convention that $0/0 = 0$.
Our analysis splits into two cases. First let us consider the \emph{small case}. This is handled by the following proposition.

\begin{proposition}\label{prop:small_case}
Suppose $\Tr(\rhohalf A\rhohalf B)<\delta/2$. Then $
	\valst(G',\{A_{x'}^{a'}\},\{B_{y'}^{b'}\})< \delta.$
\end{proposition}
\begin{proof}
First of all we have
\[
	 \valst(G',\{A_{x'}^{a'}\},\{B_{y'}^{b'}\}) = \Ex_{x, y} \sum_{V(a,b,x,y)=1} \Tr(A_x^a \rhohalf B_y^b\rhohalf) 
	 \leq  \Ex_{x,y} \Tr(A_x\rhohalf B_y\rhohalf).\]
Subtracting  $\Tr(A\rhohalf B\rhohalf)$,
\[ \Ex_{x,y} \Tr(A_x\rhohalf B_y\rhohalf)
-\Tr(A\rhohalf B\rhohalf)= \Ex_{x,y} \Tr((A_x-A) \rhohalf (B_y-B)\rhohalf).\]
By applying Cauchy-Schwarz to the latter expression and using Claim \ref{claim:expansion} it follows that
\[ \valst(G',\{A_{x'}^{a'}\},\{B_{y'}^{b'}\})\leq \delta/2+\lambda^2;\]
this is smaller than $\delta$ by the choice of $\lambda$.
\end{proof}

\paragraph{The large case.} In this case, the hypothesis of Proposition \ref{prop:small_case} is not satisfied and without loss of generality we assume that  
\begin{equation}\label{eq:large_case}
	\min\big\{\Tr_\rho(A),\, \Tr_\rho(B)\big\} \geq \Tr(A\rhohalf B\rhohalf)\geq \delta/2.
\end{equation}
Let
\begin{equation}
	\gamma := \Ex_{(x,y) \sim \mu} \Tr(A_x\rhohalf B_y\rhohalf).
\end{equation}
By the triangle inequality,
\begin{equation}\label{eq:two_bounds}
\valst(G',A_{x'},B_{y'})\leq \Ex_{(x,y) \sim \mu} |\Tr(A_x \rhohalf B_y \rhohalf ) -\gamma| + \gamma \cdot \Ex_{(x,y) \sim \mu} \sum_{V(a,b,x,y)=1} \frac{\Tr(A_x^a \rhohalf  B_y^b \rhohalf)}{\Tr(A_x\rhohalf B_y \rhohalf )}.
\end{equation}  
To bound the first term, we use the triangle inequality to get
\begin{equation}
|\Tr(A_x \rhohalf B_y \rhohalf ) -\gamma| \leq |\Tr(A_x\rhohalf B_y\rhohalf -A\rhohalf B\rhohalf )| + |\Tr(A\rhohalf B\rhohalf)-\gamma|.
\label{eq:bound-2a}
\end{equation}
The first term on the right-hand side of~\eqref{eq:bound-2a} can be bounded as
\begin{align}
\Ex_{(x,y) \sim \mu} |\Tr(A_x\rhohalf &B_y\rhohalf -A\rhohalf B\rhohalf )|\notag\\
&\leq \Ex_{(x,y) \sim \mu} |\Tr(A_x\rhohalf (B_y-B)\rhohalf) |+ \Ex_{x} |\Tr((A_x -A)\rhohalf B\rhohalf)| \notag \\
&\leq \Ex_{(x,y) \sim \mu}\left[\Tr_\rho(A_x^2)^{1/2}\cdot \Tr_\rho((B_y-B)^2)^{1/2}\right]+\Ex_{x}\left[\Tr_\rho(B^2)^{1/2}\cdot \Tr_\rho((A_x-A)^2)^{1/2}\right]\notag \\
&\leq \left(\Ex_{x}\Tr_\rho(A_x^2)\right)^{1/2}\cdot \left(\Ex_y \Tr_\rho((B_y-B)^2) \right)^{1/2}+ \Tr_\rho(B^2)^{1/2}\cdot \left(\Ex_x \Tr_\rho((A_x-A)^2) \right)^{1/2}  \notag\\
&\leq 4\cdot \lambda, 
\end{align}
where the first inequality is the triangle inequality, the next two follow from Cauchy-Schwarz, and the last from Claim~\ref{claim:expansion} and the trivial bounds $\Tr_\rho(A_x^2), \Tr_\rho(B^2)\leq \Tr(\rho)=1$. To bound the second term on the right-hand side of~\eqref{eq:bound-2a} we note that
\begin{align*}
	|\Tr(A\rhohalf B\rhohalf)-\gamma|&=|\Ex_{(x,y) \sim \mu} \Tr((A_x-A)\rhohalf(B_y-B)\rhohalf)|\\
	&\leq (\Ex_{x}\Tr_\rho[(A_x-A)^2])^{1/2}\cdot(\Ex_{y}\Tr_\rho[(B_y-B)^2])^{1/2}, 
	\end{align*}
and the latter is again bounded by $2 \lambda$ by Claim \ref{claim:expansion}. In total we have 
	\begin{equation}
		\Ex_{(x,y) \sim \mu} |\Tr(A_x \rhohalf B_y \rhohalf ) -\gamma|\leq 4\lambda+2 \lambda^2\leq \delta,
	\end{equation}
which provides an upper bound on the first term in the right-hand side of~(\ref{eq:two_bounds}). 

\medskip

To bound the second term term in the right-hand side of~(\ref{eq:two_bounds}) we use a strategy inspired in part by the parallel repetition theorem of \cite{DinurSV14}. Let $U_x, V_y, U, V$ be a family of unitaries such that the operators
\begin{equation}
\Lambda_x=U_x \sqrt{A_x} \rho^{1/4}, \quad  \Lambda=U\sqrt{A}\rho^{1/4}, \quad \Gamma_y=V_y\sqrt{B_y}\rho^{1/4}, \quad  \Gamma=V\sqrt{B}\rho^{1/4}
\end{equation}
are all positive semidefinite, which is possible by Fact~\ref{fact:PSDness}. Note that this in particular implies that $\Lambda_x=\Lambda_x^\dagger$ and hence
\begin{equation}\label{eq:square_of_x}
\Lambda_x^2= \Lambda_x^\dagger \Lambda_x= \rho^{1/4} \sqrt{A_x}U_x^\dagger U_x \sqrt{A_x}\rho^{1/4}= \rho^{1/4}A_x\rho^{1/4},
\end{equation}
and similarly $\Lambda^2=\rho^{1/4}A\rho^{1/4}$, $\Gamma^2=\rho^{1/4}B\rho^{1/4}$ and so on.

Define ``rescaled'' strategies by
\begin{equation}\label{eq:ops}
	\widehat{A_{x}^a}= U_xA_{x}^{-1/2} A_{x}^a A_{x}^{-1/2}U_x^\dagger, \quad \widehat{B_{y}^b}= V_yB_{y}^{-1/2} B_{y}^b B_{y}^{-1/2}V_y^\dagger,
\end{equation}
where $A_{x}^{-1}, B_y^{-1}$'s are the pseudo-inverses of $A_{x},B_y$.
Note that the operators~\eqref{eq:ops} satisfy $\widehat{A_x}=\sum_{a} \widehat{A_x^a}$, $\widehat{B_y}=\sum_b \widehat{B_y^b}\leq \Id$ as required. Let
\begin{equation}
K_{xy}= \frac{U_{x}A_{x}^{1/2}\rho^{1/2}B_y^{1/2}V_y^\dagger}{\sqrt{\Tr(A_x\rhohalf B_y\rhohalf)}}, \qquad K= \frac{UA^{1/2}\rho^{1/2}B^{1/2}V^\dagger}{\sqrt{\Tr(A\rhohalf B\rhohalf)}}.
\end{equation}
By definition of $\Lambda_x, \Gamma_y, X, Y$ we see that the above is equivalent to
\begin{equation}
K_{xy}= \frac{\Lambda_x\Gamma_y}{\sqrt{\Tr(\Lambda_x^2\Gamma_y^2)}}, \qquad K= \frac{\Lambda \Gamma}{\sqrt{\Tr(\Lambda^2\Gamma^2)}}.
\end{equation}
Now note the following identity 
\begin{equation}
\frac{\Tr(A_x^a\rhohalf B_y^b\rhohalf)}{\Tr(A_x\rhohalf B_y\rhohalf )}=  \Tr(\widehat{A_x^a}K_{xy} \widehat{B_y^b}K_{xy}^\dagger).
\end{equation}
So to finish the argument it suffices to estimate
\begin{equation}
\Ex_{(x,y) \sim \mu} \sum_{V(a,b,x,y)=1} \Tr(\widehat{A_x^a}K_{xy} \widehat{B_y^b}K_{xy}^\dagger).
\end{equation}
To this end note that since $\Tr(KK^\dagger)=1$ it follows from the definition of $\valst(G)$ that 
\begin{equation}\label{eq:main_term_goal}
\Ex_{(x,y) \sim \mu} \sum_{V(a,b,x,y)=1} \Tr(K \widehat{A_x^a} K^\dagger \widehat{B_y^b})\leq \valst(G).
\end{equation}
To conclude we use the following proposition.

\begin{proposition}\label{prop:technical_bound}
Let $K_{xy}$ and $K$ be as above. Then 
\begin{equation}
\Ex_{(x,y) \sim \mu} \|K_{xy}-K\|_{F}^2\leq \frac{12\lambda}{\delta}.	
\end{equation}
\end{proposition}
Before proving the proposition let us see how it implies the desired bound on the second term of \eqref{eq:two_bounds}.
\begin{align}
|\Tr(K_{xy} \widehat{A_x^a} K_{xy}^\dagger& \widehat{B_y^b})-	\Tr(K \widehat{A_x^a} K^\dagger \widehat{B_y^b})| \notag\\
&\leq |\Tr( (K_{xy}-K) \widehat{A_x^a} K_{xy}^\dagger \widehat{B_y^b})|+ |\Tr(K\widehat{A_x^a} (K_{xy}^\dagger-K^\dagger)\widehat{B_y^b})	| \notag \\
&\leq \Tr( (K_{xy}-K)\widehat{A_x^a} (K_{xy}-K)^\dagger \widehat{B_y^b})^{1/2}\cdot \Tr( K_{xy}\widehat{A_x^a} K_{xy}^\dagger \widehat{B_y^b})^{1/2} \notag
\\  &\qquad+\Tr( (K_{xy}-K)\widehat{A_x^a} (K_{xy}-K)^\dagger \widehat{B_y^b})^{1/2} \cdot \Tr( K\widehat{A_x^a} K^\dagger \widehat{B_y^b})^{1/2}. \label{eq:prop_long_3}
\end{align}
Averaging with $\Ex_{(x,y) \sim \mu}\sum_{V(a,b,x,y)=1}$ and applying  Cauchy-Schwarz we see that~\eqref{eq:prop_long_3} is bounded by
\begin{align}
\notag	 \Big[\Ex_{(x,y) \sim \mu}\sum_{V(a,b,x,y)=1} \Tr( (K_{xy}-K)\widehat{A_x^a} (K_{xy}-K)^\dagger \widehat{B_y^b})\Big]^{1/2}\cdot &\Big[\big(\Ex_{(x,y) \sim \mu}\sum_{V(a,b,x,y)=1}\Tr( K_{xy}\widehat{A_x^a} K_{xy}^\dagger \widehat{B_y^b}) \big)^{1/2} \notag \\
	 &+\big(\Ex_{(x,y) \sim \mu}\sum_{V(a,b,x,y)=1}\Tr( K\widehat{A_x^a} K^\dagger \widehat{B_y^b})\big)^{1/2}\Big].\label{eq:prop_long_3b}
\end{align}
We claim that the second term in brackets is at most $2$. To see this note that replacing the sum from $\sum_{V(a,b,x,y)=1}$ to a $\sum_{a,b}$ only increase the term, and the claim follows from $\Tr(K_{xy}K_{xy}^\dagger)=\Tr(KK^\dagger)=1$. To bound the first term in~\eqref{eq:prop_long_3b}, we again relax the summation from  $\sum_{V(a,b,x,y)=1}$ to $\sum_{a,b}$. This is valid because all the operators of the form $(K_{xy}-K)\widehat{A_x^a} (K_{xy}-K)^\dagger, B_y \geq 0$ and hence all the additional terms introduced in the sum are nonnegative. The desired result follows because 
\begin{equation}
	\Ex_{(x,y) \sim \mu} \sum_{a,b} \Tr( (K_{xy}-K)\widehat{A_x^a} (K_{xy}-K)^\dagger \widehat{B_y^b})\leq \Ex_{(x,y) \sim \mu}\Tr((K_{x,y}-K)(K_{x,y}-K)^\dagger)=\Ex_{x,y} \|K_{x,y}-K\|_{F}^2,
\end{equation}
which is bounded by Proposition \ref{prop:technical_bound}. Combining all bounds, from~(\ref{eq:prop_long_3}) we get
\begin{align*}
	\Ex_{(x,y) \sim \mu} \sum_{V(a,b,x,y)=1} \frac{\Tr(A_x^a \rhohalf  B_y^b \rhohalf)}{\Tr(A_x\rhohalf B_y \rhohalf )} &=\Ex_{(x,y) \sim \mu} \sum_{V(a,b,x,y)=1} \Tr(\widehat{A_x^a}K_{xy} \widehat{B_y^b}K_{xy}^\dagger)  \\
	&\leq \Ex_{(x,y) \sim \mu}  \sum_{V(a,b,x,y)=1} \Tr(\widehat{A_x^a}K \widehat{B_y^b}K^\dagger) \\
	&\qquad+|\Tr(K_{xy} \widehat{A_x^a} K_{xy}^\dagger \widehat{B_y^b})-	\Tr(K \widehat{A_x^a} K^\dagger \widehat{B_y^b})|  \\
	&\leq \valst(G)+ 2\cdot \left(\Ex_{(x,y) \sim \mu} \|K_{xy}-K\|_{F}^2 \right)^{1/2} \\
	&\leq \valst(G)+ 2\sqrt{\frac{12\lambda}{\delta}}.
\end{align*}
The latter is bounded by $\eps$ by the choice of $\lambda$. It only remains to prove Proposition \ref{prop:technical_bound}.
\begin{proof}[Proof of Proposition \ref{prop:technical_bound}]
	We have 
	\begin{equation}\label{eq:main_prop_long}
	\|K_{xy}-K\|_{F} \leq \bigg\|\frac{\Lambda_x\Gamma_y}{\sqrt{\Tr(\Lambda_x^2\Gamma_y^2)}}- \frac{\Lambda\Gamma}{\sqrt{\Tr(\Lambda^2\Gamma^2)}} \bigg\|_{F} +  \bigg\|\frac{\Lambda_x\Gamma_y}{\sqrt{\Tr(\Lambda^2\Gamma^2)}}- \frac{\Lambda\Gamma}{\sqrt{\Tr(\Lambda^2\Gamma^2)}} \bigg\|_{F}.
\end{equation}
For the first term, 
\begin{align}
	\Ex_{(x,y) \sim \mu} \bigg\|\frac{\Lambda_x \Gamma_y}{\sqrt{\Tr(\Lambda_x^2\Gamma_y^2)}}- \frac{\Lambda_x\Gamma_y}{\sqrt{\Tr(\Lambda^2\Gamma^2)}} \bigg\|_{F}^2 &= \Ex_{(x,y)\sim \mu} \Tr(\Lambda_x^2\Gamma_y^2)\cdot \bigg(\frac{1}{\sqrt{\Tr(\Lambda_x^2\Gamma_y^2)}}- \frac{1}{\sqrt{\Tr(\Lambda^2\Gamma^2)}}\bigg)^2 \notag\\
	& =\frac{1}{\Tr(\Lambda^2\Gamma^2)} \Ex_{(x,y) \sim \mu} \bigg( \sqrt{\Tr(\Lambda_x^2\Gamma_y^2)}- \sqrt{\Tr(\Lambda^2\Gamma^2)}\bigg)^2  \notag \\
	 &\leq  \frac{1}{\Tr(\Lambda^2\Gamma^2)} \Ex_{(x,y) \sim \mu} |\Tr(\Lambda_x^2\Gamma_y^2)-\Tr(\Lambda^2\Gamma^2)| \notag \\
	 &\leq \frac{1}{\Tr(\Lambda^2\Gamma^2)} \Ex_{(x,y) \sim \mu} |\Tr( (\Lambda_x^2-\Lambda^2) \Gamma_y^2)| + |\Tr(\Lambda^2(\Gamma_y^2-Y^2))|   \notag \\
	 & \leq\frac{1}{\Tr(\Lambda^2\Gamma^2)} \left|\Ex_{x}[\Tr( (\Lambda_x^2-\Lambda^2)^2)]\right|^{1/2}\cdot \left|\Ex_{y}[\Tr(\Gamma_y^4)] \right|^{1/2} \label{eq:first_term_prop_long1}\\  
	 &+ \label{eq:first_term_prop_long2} \frac{1}{\Tr(\Lambda^2\Gamma^2)} \left|\Ex_{x}[\Tr( (\Gamma_y^2-\Gamma^2)^2)]\right|^{1/2}\cdot \Tr(\Lambda^4)^{1/2}, 
\end{align}
where  the last step follows from two applications of Cauchy-Schwarz. Rewriting the above in terms of $A_x, B_y$ and $\rho$ using~(\ref{eq:square_of_x}) and its analogues we see that the term in~(\ref{eq:first_term_prop_long1}) equals
\begin{equation}
	\frac{1}{\Tr(A\rhohalf B\rhohalf)}\left|\Ex_{x}[\Tr((A_x-A)\rho^{1/2} (A_x-A)\rho^{1/2})]\right|^{1/2}\cdot \left|\Ex_{y}[\Tr(B_y\rhohalf B_y \rhohalf)\right|^{1/2} 
\end{equation}
Bounding  the last term $\Tr(B_y\rhohalf B_y \rhohalf)$ by $1$ and the first term by $2\lambda$ (which follows by applying Fact \ref{fact:cor_cs} and Claim \ref{claim:expansion}) and doing the same analysis for (\ref{eq:first_term_prop_long2}) we see that
\begin{equation}
	\Ex_{(x,y) \sim \mu} \bigg\|\frac{\Lambda_x \Gamma_y}{\sqrt{\Tr(\Lambda_x^2\Gamma_y^2)}}- \frac{\Lambda_x\Gamma_y}{\sqrt{\Tr(\Lambda^2\Gamma^2)}} \bigg\|_{F}^2 \leq \frac{8\lambda}{\delta}.
\end{equation}
To bound the second term in (\ref{eq:main_prop_long}) we argue as follows:
\begin{align} \notag
\| \Lambda_x \Gamma_y - \Lambda \Gamma\|_F^2 &\leq 2\cdot\|(\Lambda_x-\Lambda)\Gamma_y\|_F^2+2\cdot \|(\Gamma_y- \Gamma)X\|_F^2 \\
&= 2\cdot \Tr(\Gamma_y^2 (\Lambda_x-\Lambda)^2) + 2\cdot \Tr((\Gamma_y-\Gamma)^2 X^2)  \\
&\leq 2\cdot \Tr(\Gamma_y^4)^{1/2} \cdot \Tr( (\Lambda_x-\Lambda)^4)^{1/2}  +2 \cdot \Tr(\Lambda^4)^{1/2}  \cdot \Tr((\Gamma_y-\Gamma)^4)^{1/2}  \\
&\leq 2\cdot \Tr(\Gamma_y^4)^{1/2}\cdot \Tr[(\Lambda_x^2-\Lambda^2)^2]^{1/2} +2 \cdot \Tr(\Lambda^4)^{1/2}\cdot \Tr[ (\Gamma_y^2-\Gamma^2)^2]^{1/2}, 
\end{align}
where in the last step we used Lemma~\ref{lemma:kit}. Using the same bound on the above terms as in the above we see that 
\begin{equation}
\Ex_{(x,y) \sim \mu} \| \Lambda_x \Gamma_y - \Lambda \Gamma\|_F^2 \leq 8\lambda.
\end{equation}
Since in the large case $\Tr(A\rhohalf B\rhohalf)\geq \frac{\delta}{2}$ the result follows. 
\end{proof}

\section{Discussion and open problems}\label{sec:open_problems}
An obvious open problem is to extend our results to the case of multiplayer entangled games. This is likely to be achievable by combining the ideas of Sections \ref{sec:multiplaer_classical} and~\ref{sec:quantum_generalization}. However, some subtleties arise with respect to the use of the Choi-Jamiolkowski isomorphism, and we leave this for future work. 

An important (and somewhat surprising) message of our work is that there is a modified form of game concatenation with no adverse effect on the entangled value. This is notable because ordinary concatenation may appear to be not very well-behaved with respect to quantum strategies: we typically do not expect that entangled players would be able to answer a number of questions from a game $G$ simultaneously, while preserving the same question/answer statistics as in $G$, as players' measurement operators associated with different questions generally do not commute. 

The concatenation and composition of games play an important role in the classical setting in the context of PCPs \cite{arora_safra,dinur2007pcp} and multiprover interactive  proof systems \cite{babai1990}. It remains to be seen whether ideas related to our ordered fortification can be useful in lifting some of these techniques to the quantum world.

\bibliography{fortbi}
\appendix

\section{Biregularization}\label{app:bireg}

In this section, we prove Lemmas \ref{lem:biregular} and \ref{lem:biregular_graphical}. We start by the second lemma on the graphical games and then derive the general case by reduction. 

\paragraph{Graphical games.}  Suppose $G$ is a graphical game: there is a set of edges $E\subseteq X\times Y$ such that $\mu(x,y)=\frac{1}{|E|}$ for all $(x,y)\in E$. In this case we have 
\begin{equation}
\mu(x)=\frac{|N(x)|}{|E|}, \qquad \mu(y)= \frac{|N(y)|}{|E|}, \qquad \forall x\in X,\, y\in Y.	
\end{equation}
Let $d_x:= |N(x)|$ denote the degree of $x$, and set $S_x$ be the set $\{x\}\times [d_x]$. Define
\[\Xint = \bigcup_{x\sim X} S_x.\]
Note that $|\Xint | =|E| \leq |X| |Y|$. Define $M_{int}((x,i), x)=\frac{1}{d_x}$ for $i\in\{1,\ldots, d_x\}$, and $0$ otherwise. Construct $\Yint $ and $P_{int}$ similarly. 
\begin{proposition}
Let $G_{int}= M_{int} \circ G \circ P_{int}$. Then marginal of $\mu_{int}$ induced on $\Xint $ and $\Yint $ is uniform. Moreover, we have 
\begin{equation}
	 \val(G_{int})=\val(G), \qquad \val^*(G_{int})= \val^*(G). 
\end{equation}
\end{proposition}
\begin{proof}
 It is easy to see that for all $\Xint =(x,i) \in \Xint $ we have $\mu_{int}(x_{int} )= \frac{1}{|E|}$ and similarly for all $y_{int}  \in \Yint $. The claims $\val(G_{int})=\val(G)$ and $\val^*(G_{int}) \leq \val^*(G)$ are true for all concatenated games in general. The final claim $\val^*(G_{int})\geq \val^*(G)$ follows by considering the strategy $A_{(x,i)}= A_x$, $B_{(y,j)}=B_y$  which achieves the same value as $(A_x, B_y)$ in $G$. 
\end{proof}

\paragraph{General case.} Although graphical games  include many games considered in applications, it would nevertheless still be nice to extend the above construction to all games. We do not know how to do this exactly, but we can achieve an approximate variant.

The idea is essentially to approximate a general game by a graphical game. More formally, let $\tau \in (0,1)$ be an error parameter and $q$ an integer such that $\frac{|E|}{\tau}\leq q\leq  \frac{2|E|}{\tau}$. We have 
\begin{equation}
\frac{\tau}{2|E|} \leq \frac{1}{q} \leq \frac{\tau}{|E|}. 	
\end{equation}

We would like to define a game  $\tilde{G}$ in which all probabilities in the underlying distribution $\tilde{\mu}(x,y)$ are fractions with denominator $q$. Let $\tilde{X}=X\cup \{x_{nul}\}$ and $\tilde{Y}= Y\cup \{y_{nul}\}$. For every $(x,y)\in X\times Y$ set
\begin{equation}
\tilde{\mu}(x,y)= \frac{ \lfloor q\cdot \mu(x,y) \rfloor}{q}.
\end{equation}
Finally let $\tilde{\mu}(x_{nul}, y_{nul})$ such that $\tilde{\mu}$ is a proper probability distribution (i.e.~by transferring the excess probabilities to $(x_{nul}, y_{nul})$) and put an arbitrary winnable predicate on $(x_{nul}, y_{nul})$. 

\begin{proposition}
The game $\tilde{G}$ is a graphical game with $q$ (possibly parallel) edges. Moreover, we have 
\begin{equation}
\val(G) \leq \val(\tilde{G})\leq \val(G)+\tau, \qquad \val^{*}(G)\leq \val^* (\tilde{G})\leq \val^* (G)+ \tau. 
\end{equation}
\end{proposition}
A few remarks are in order: firstly, since the previous construction for graphical games applies equally well in the presence of multiples edges, we can combine it with the above preprocessing to prove Lemma  \ref{lem:biregular}. Secondly, note that the operation $G\rightarrow \tilde{G}$ is value-increasing and hence preserves perfect completeness. Thirdly, note that the right scale for the error parameter $\tau$ is $\frac{c-s}{2}$ where $c-s$ is the completeness-soundness gap.  
\begin{proof}
By construction all $\tilde{\mu}(x,y)$ are integer multiples of $\frac{1}{q}$. This ensures that the same is true for $\tilde{\mu}(x_{nul}, y_{nul})$. Since $\mu(x,y)\geq \tilde{\mu}(x,y)$ for all $(x,y)$, for any strategy $(f,g)$ for $G$ we have 
\[ 1-\val(G,f,g)= \Ex_{(x,y)\sim \mu} \sum_{V(a,b,x,y)=0} f(x,a)\cdot g(y,b)\geq 1-\val(\tilde{G}, f,g),\]
which shows that $\val(G)\leq \val(\tilde{G})$. For the other direction, consider an optimal strategy $(f,g)$ for $\tilde{G}$ (which necessarily always wins on $(x_{nul}, y_{nul})$). We have,
\begin{align*}
1-\val(G) & \leq 1-\val(G,f,g)= 	\Ex_{(x,y)\sim \mu} \sum_{V(a,b,x,y)=0} f(x,a)\cdot g(y,b) \\
& \leq \sum_{x,y} \tilde{\mu}(x,y) \sum_{V(a,b,x,y)=0} f(x,a)\cdot g(y,b)+  \sum_{(x,y)\in E} \left(\mu(x,y)- \tilde{\mu}(x,y)\right) \\
&\leq 1-\val(\tilde{G}) +\tau
\end{align*}
The quantum case is similar.
\end{proof}

\end{document}